\documentclass[a4paper,twocolumn,11pt]{quantumarticle}
\pdfoutput=1

\usepackage[utf8]{inputenc}
\usepackage[english]{babel}
\usepackage[T1]{fontenc}
\usepackage{amsmath, amssymb, amsfonts, amsthm}
\usepackage{bm, bbm, physics, mathtools}
\usepackage{hyperref}

\usepackage{tikz}
\usepackage{lipsum}

\usepackage{algorithm}
\usepackage[noend]{algpseudocode}
\usepackage{thm-restate}

\hypersetup{
  colorlinks=true,
  linkcolor=blue,
  citecolor=blue,
  urlcolor=blue
}

\newtheorem{theorem}{Theorem}
\newtheorem{proposition}[theorem]{Proposition}

\newtheorem{definition}[theorem]{Definition}

\renewcommand{\eqref}[1]{Eq.~(\ref{#1})}
\newcommand{\figref}[1]{Fig.~\ref{#1}}
\newcommand{\secref}[1]{Sec.~\ref{#1}}

\renewcommand{\rank}{{\rm rank}}

\renewcommand{\deg}{{\rm deg}\hspace{2pt}}
\newcommand{\wt}{{\rm wt}\hspace{2pt}}

\def\NN{{\mathbb N}}
\def\ZZ{{\mathbb Z}}

\def\FF{{\mathbb F}}

\def\bmc{{\bm c}}
\def\bme{{\bm e}}
\def\bmo{{\bm 0}}

\def\bmr{{\bm r}}
\def\bms{{\bm s}}
\def\bmO{{\bm O}}

\def\C{{\mathcal C}}
\def\F{{\mathcal F}}
\def\H{{\mathcal H}}

\def\O{{\mathcal O}}
\def\P{{\mathcal P}}
\def\Q{{\mathcal Q}}
\def\S{{\mathcal S}}

\begin{document}

\title{Small Quantum Codes from Algebraic Extensions of Generalized Bicycle Codes}
\author{Nikolaos Koukoulekidis}
    \affiliation{IQM, Georg-Brauchle-Ring 23-25, 80992 Munich, Germany}
    \affiliation{Dept. of Electrical and Computer Engineering, Duke University, Durham, NC, USA}
    \affiliation{Dept. of Physics, Duke University, Durham, NC, USA}
    \email{nikolaos.koukoulekidis@duke.edu}
\author{Fedor \v{S}imkovic IV}
    \affiliation{IQM, Georg-Brauchle-Ring 23-25, 80992 Munich, Germany}
\author{Martin Leib}
    \affiliation{IQM, Georg-Brauchle-Ring 23-25, 80992 Munich, Germany}
\author{Francisco Revson F. Pereira}
    \affiliation{IQM, Georg-Brauchle-Ring 23-25, 80992 Munich, Germany}
    \email{francisco.revson@meetiqm.com}

\begin{abstract}
Quantum error correction is rapidly seeing first experimental implementations, but there is a significant gap between asymptotically optimal error-correcting codes and codes that are experimentally feasible.
Quantum LDPC codes range from the surface code, which has a vanishing encoding rate, to very promising codes with constant encoding rate and linear distance. 
In this work, motivated by current small-scale experimental quantum processing units, we devise small quantum codes that are inspired by a subset of quantum LDPC codes, known as generalized bicycle (GB) codes.
We introduce a code construction based on algebraic manipulation of the parity-check matrix of GB codes, rather than manipulation of Tanner graphs.
Our construction leads to families of quantum LDPC codes of small size, and we demonstrate numerically that their performance scales comparably to the performance of surface codes for similar sizes under a phenomenological noise model. 
The advantage of our code family is that they encode many logical qubits in one code, at the expense of non-local connectivity.
We then explore three variants of the code construction focusing on reducing the long-range connectivity by bringing it closer to the current experimental capabilities of short-range connectivity devices.
\end{abstract}
\maketitle

\section{Introduction}

Quantum information is prone to destructive errors while stored or manipulated.
The number of errors increases with the number of components in a quantum circuit.
Therefore, active quantum error correction (QEC) is deemed essential to realize scalable quantum computer architectures~\cite{Shor1995,terhal2015,devitt2013quantum}.
A QEC code encodes quantum information in a larger space in order to introduce redundancy, so that errors on few physical qubits do not alter logical information irreversibly.
Errors on the logical state are detected via a series of non-destructive, projective measurements, called parity checks, that yield syndrome bits.
These syndrome bits are processed via a decoder, tasked with determining the optimal recovery operation which restores the error-free logical state.
Therefore, apart from data qubits that encode computational information, one requires auxiliary, syndrome qubits to perform the parity checks.
The full sequence of syndrome measurements, syndrome decoding and recovery comprise a QEC cycle and requires substantial space-time computational overheads~\cite{aharonov1999,terhal2015}.

In classical coding theory, low density parity-check (LDPC) codes~\cite{gallager1963low} are codes with low overhead, as the number of data bits per parity-check and the number of parity-checks that a data bit participates in are both bounded by a constant, independently of the scaling of the code.
This is represented by a sparse parity-check matrix. 
A code can generally be defined via a parity-check matrix, where each row represents a syndrome bit and each column represents a data bit.
Equivalently, it can be represented by a bipartite ``factor'' graph~\cite{Tanner1981} which illustrates the connections as edges between vertices that represent syndrome and data bits.
LDPC codes form very successful error-correction protocols as they satisfy the sparsity condition and they also saturate upper bounds on the amount of information that can be reliably transferred through noisy channels~\cite{gallager1963low,mackay1997near}.
As a result, they constitute the industrial standard for many modern technologies such as the WiFi and 5G~\cite{bae2019overview,iliev2008application}.

It is natural to look for quantum analogs to LDPC codes, a task that recently has attracted a lot of attention, with many families of quantum LDPC codes~\cite{gottesman2014,Kovalev2013fault,Breuckmann2021} being developed.
There are different methods of constructing quantum LDPC codes.
One approach is to generate a quantum code with parity-check matrix $H$ via the hypergraph product construction~\cite{tillich2013quantum} which allows any classical code to be converted to a quantum code.
This approach permits construction of random quantum LDPC codes, which may generally possess favorable properties but also require unrestricted qubit connectivity and, crucially, there is no methodical way of designing larger codes from smaller ones, as their parity-check matrices are unrelated.

On the other hand, topological codes~\cite{dennis2002topological,freedman2001projective,Bombin2006}, such as the surface and color codes, are quantum LDPC codes that arise as the hypergraph product of classical repetition codes.
They possess some of the highest known noise thresholds, and are local in the sense that all parity checks can be performed via entangling gates between syndrome qubits and nearest-neighbor data qubits~\cite{fowler2012surface}.
However, they provide a vanishing logical qubit encoding rate as the code length increases.
Semi-topological codes~\cite{roffe_decoding_2020} have been constructed via a process that augments the edges on the factor graph of the underlying classical LDPC code that defines a given quantum LDPC code via the hypergraph product.
This allows for a trade-off between code threshold and parity-check locality.

Recently, quantum LDPC codes have demonstrated rapid success in the asymptotic regime. 
Good quantum LDPC codes of code length $n$ with a non-vanishing encoding rate and a minimum distance $d$ linear in $n$ have been discovered~\cite{panteleev2022asymptotically,leverrier2022quantum}.
It is well-known that these good properties come at the expense of a number of long-range interactions~\cite{Bravyi2010tradeoffs,Baspin2022,Tremblay2022}.
Significant progress has been made in providing remedies for the QEC cycle of quantum LDPC codes in terms of guaranteeing the existence of a syndrome-extraction circuit~\cite{delfosse2021bounds,pattison2023hierarchical} and fast decoders~\cite{leverrier2022decoding}.
However, these asymptotically good constructions do not yet guide us towards quantum LDPC codes with small code lengths.

Regarding the existence of quantum LDPC codes performing significantly better than the surface code under small code lengths $n<1000$, Panteleev and Kalachev showed that \emph{generalized bicycle (GB)} codes, initially introduced by Kovalev and Pryadko~\cite{Kovalev2013quantum}, can give quantum LDPC codes with such properties~\cite{Panteleev2021degenerate}. Quantum LDPC codes derived from GB codes preserve a repeated check structure under the expense of introducing non-local checks. Excitingly, recent papers have used GB codes to design quantum LDPC codes that can be implemented on superconducting qubits and atom arrays hardware~\cite{Tremblay2022,Xu2023,Bravyi2023,Viszlai2023}. In particular, Bravyi,~\emph{et al.}~\cite{Bravyi2023} showed the existence of quantum LDPC codes requiring data and syndrome qubit connectivity equal to six and having a Tanner graph with thickness two, meaning that it could be implemented using flip-chip technology with two planar layers of couplers and control lines on the chip hosting data and syndrome qubits. Lastly, 
it has been recently demonstrated by Viszlai,~\emph{et al.}~\cite{Viszlai2023} that the repeated check structure of GB codes matches with
the ability to move on $2D$ grids of qubits in atom arrays. Additionally, they consider a mixed architecture where quantum LDPC codes derived from GB codes are used as quantum memory and surface codes are used for computation. This reveals the strength of each of these two codes can be combined to obtain improved performance.

We are interested in small scale quantum LDPC codes.
To this end, we introduce a code construction based on manipulating the underlying algebraic structure of GB codes~\cite{mackay2004sparse,Kovalev2013quantum,Panteleev2021degenerate}.
Specifically, the parity-check matrix of GB codes is defined in terms of two circulant matrices and we develop algebraic transformations that extend these circulant matrices, creating a family of codes with increasing code length.
We present a family of quantum LDPC $[[n_m,k_m]]$ codes, indexed by $m \geq 1$, with small code lengths $n_m = 10m$ and dimension $k_m \geq 2$.
We demonstrate a code in this family with similar size to the surface code of distance 7, where size is counted as the sum of data and syndrome qubits, and similar threshold performance.
The advantage of this quantum LDPC code over the surface code is that it encodes 2 logical qubits instead of 1.
However, our quantum LDPC codes require longer-range connectivity as $m$ increases.

We then explore the experimental scalability of such code extensions, inspired by platforms which can support long-range connectivity at a small scale, such as architectures based on superconducting qubits and co-planar waveguide resonators~\cite{Algaba2022}, ion traps~\cite{monroe2013scaling}, photons~\cite{rudolph2017optimistic} or Rydberg atoms~\cite{morgado2021quantum,Bluvstein2023}.
In particular, we use our construction methodology to find a new family of quantum codes $[[n_m,k_m]]$, with $k_{m+1} \geq k_m$ subject to the scalability
property that the parity-check matrices of these codes consist of copies of a single small building block. The building block represents an arbitrary connectivity graph that is experimentally easy to obtain.
This property resembles the construction of topological codes, but we allow the building block to be derived by any GB code, hence the code dimension can be much larger than 1.
Therefore, our construction allows for a procedural design of chips which are composed of regions that locally shares the same connectivity structure of the smaller codes in the same family. It is worth noting that the chip may require some non-planar / long-range connectivity between the building blocks.

The rest of the manuscript is divided as follows.
In \secref{sec:prelim}, we introduce notation and preliminary definitions on the quantum codes considered.
We then introduce our algebraic construction in \secref{sec:construction}, and we obtain example of families of quantum LDPC codes as an application of this construction.
In \secref{sec:numerics}, we numerically investigate the performance of the smallest of the newly introduced quantum LDPC codes.
In \secref{sec:scalable}, we explore the experimental scalability of codes arising from our construction.
Finally, we conclude and discuss our work in \secref{sec:summary}.

\section{Background and preliminaries}
\label{sec:prelim}

\subsection{Quantum Stabilizer Codes}

Consider an $n$--qubit Hilbert space $\H \coloneqq \mathbb{C}^{2^n} \cong (\mathbb{C}^{2})^{\otimes n}$.
The set of $n$--qubit Pauli operators $\P_n$ consists of all operators $P = \alpha P_1 \otimes \dots \otimes P_n$, where $\alpha \in \{\pm 1, \pm i\}$ and $P_i \in \{I,X,Y,Z\}$ for $i=1,\dots,n$.
The weight $w$ of a Pauli operator $P$ is the number of non-identity components $P_i$ in the tensor product.
The Pauli group modulo the phase factor $\alpha$, i.e. $\P_n/\{\pm I, \pm iI\}$, is isomorphic to the commutative group $\ZZ_2^{2n}$ of $2n$--binary strings, with the isomorphism given by
\begin{eqnarray}\label{eq:pauli2binary}
    \bigotimes_{i=1}^n P_i &=& \bigotimes_{i=1}^n X^{x_i} Z^{z_i}\nonumber\\
    &\mapsto& (x_1, \dots, x_n | z_1, \dots, z_n) = (x|z).
\end{eqnarray}
It is then known that two Pauli operators $P,P' \in \P_n$ commute if and only if the following condition on their respective binary representations $(x,z), (x',z') \in \ZZ_2^{2n}$ is satisfied,
\begin{equation}\label{eq:commutativity2binary}
    [P,P']=0 \Leftrightarrow x \cdot z' + z \cdot x' = 0,
\end{equation}
where $a \cdot b = \sum_{i=1}^n a_i b_i$ is the inner product on $\ZZ_2^{2n}$.

An $[[n,k,d]]$ quantum stabilizer code $\Q$ is defined by a stabilizer group $\S \subseteq \P_n$, generated by $n-k$ elements $\langle S_1, \dots, S_{n-k} \rangle$, called parity checks, such that the code corresponds to the $2^k$--dimensional $(+1)$--eigenspace of $\S$,
\begin{equation}
    \Q \coloneqq \big\{\ket{\psi} \in \mathbb{C}^{2^n} | S_i \ket{\psi} = \ket{\psi}, i = 1, \dots, n-k \big\}.
\end{equation}
The minimum distance $d$ of the code indicates the number of physical errors required to cause a logical error. Formally, it is defined as the minimum weight of a Pauli operator $P$ such that $[P,Q] = 0$ for all $Q \in \S$, but $P \notin \S$. 
To identify the presence of an error, we measure the elements in a generating set of the stabilizer group. Observe that one can also measure $\ell\geq n-k$ elements of the stabilizer group (not necessarily independent) and use the redundant information to improve the decoding performance. Henceforth, we are going to consider the later case, where $\ell \geq n-k$ stabilizer measurements are performed.

Applying the isomorphic mapping of \eqref{eq:pauli2binary} to the stabilizer generators $S_1, \dots, S_{\ell}$, we obtain an $\ell \times 2n$ binary matrix $H = (H_X | H_Z)$, called the parity-check matrix of $\Q$, where the rows of the parity-check matrix correspond to the $\ell$ stabilizer generators.
According to \eqref{eq:commutativity2binary}, the nullspace of $H$ consists of vectors $(z|x)$ such that the representations $(x|z)$ correspond to all the Pauli operators that commute with the stabilizer generators.
In particular, all stabilizer generators need to commute with themselves, thus resulting in the commutativity condition $H_X H_Z^T + H_Z H_X^T = 0$ that any stabilizer code needs to satisfy.
By the rank-nullity theorem, the code dimension can be calculated as $k = n - \rank(H)$, since this is the number of linearly independent codewords that $H$ permits.

An important class of stabilizer codes is the CSS codes~\cite{Calderbank1996,steane1996}, defined by the property that for any stabilizer generator, the non-identity components in its tensor product representation are either all $X$ or all $Z$.
In practical terms, this ensures that $X$ and $Z$ parity check measurements can be performed independently during the error correction process.
This is reflected by the fact that the parity-check matrix can be represented as
\begin{align}
    H = \left( \begin{array}{c|c} H_X & 0 \\ 0 & H_Z \\ \end{array}\right),
\end{align}
where $H_X, H_Z$ are binary matrices with $n$ columns.
The commutativity condition for CSS codes therefore reads
\begin{align}\label{eq:commutativity}
    H_X H_Z^T = 0,
\end{align}
and the code dimension can be calculated as
\begin{align}
    k - \rank(H_X) - \rank(H_Z).
\end{align}

\subsection{Quantum LDPC codes}

A classical LDPC code $\C$ can be thought of as the nullspace of a sparse $(\ell \times n)$ parity-check matrix $H_{\C}$, so that every $n$--bit codeword $\bmc \in \C$ if and only if $H_{\C} \cdot \bmc^T = \bmo$.
The sparsity requirement for LDPC codes is reflected by a constant number of non-zero elements in the parity-check matrix as its size increases, i.e.~as one considers a family of codes with parity-check matrices $H_{\C}$ of increasing $\ell$ and $n$.
Constructing classical LDPC codes can be as easy as sampling random parity-check matrices.
In fact, this method gives classical codes with good asymptotic properties $k \in \Theta(n)$ and $d \in \Theta(n)$~\cite{richardson2008modern}.
However, this simple approach generally fails for quantum LDPC codes because the commutativity constraint of \eqref{eq:commutativity} fails with high probability.

Analogously, quantum LDPC (qLDPC) codes are stabilizer CSS codes with sparse parity-check matrices.
Hereon, we assume that $H_X$ and $H_Z$ are both $(\ell \times n)$ matrices, so $H$ is a $(2\ell \times 2n)$ parity-check matrix.
We first define the row and column Hamming weights for $H$, which will allow us to specify a sparsity condition.
The Hamming weight of row $i$, 
\begin{equation}
    w^{(i)}_{\rm r} \coloneqq \sum_{j=1}^{2n} H_{ij}
\end{equation}
indicates the number of qubits in the $i$-th parity check.
The Hamming weight of column $j$, 
\begin{equation}
    w^{(j)}_{\rm c} \coloneqq \sum_{i=1}^{2\ell} H_{ij}
\end{equation}
indicates the number of $X$ parity checks that the $j$-th qubit participates in if $1 \leq j \leq n$ and the number of $Z$ parity checks that the $(j-n)$-th qubit participates in if $n+1 \leq j \leq 2n$.
We set $w_{\rm r} \coloneqq \max_{i}{w^{(i)}_{\rm r}}$ and $w_{\rm c} \coloneqq \max_{j}{w^{(j)}_{\rm c}}$.

Given a family of quantum CSS codes $\F \coloneqq \{\Q_1, \Q_2,\dots\}$, we can identify a code $\Q_m$ as a \emph{$q(m)$--sparse code} if the density of qubits per parity check and parity checks per qubit scale at most as $q(m)$.
To formalize this, we recall that for two positive sequences $f$ and $g$, we write $f(n) = \O(g(n))$ if and only if there exist positive constants $c, n_0$ such that $f(n) \leq cg(n)$ for all integers $n \geq n_0$.
We then provide the following definition.
\begin{definition}\label{def:sparsity}
    Let $\F \coloneqq \{\Q_1, \Q_2,\dots\}$ be a family of quantum CSS codes, where $\Q_m$ is defined for $m=1,2,\dots$ via a $(2\ell_m \times 2n_m)$ parity-check matrix $H_m$, with maximum row weight $w_{\rm r}(m)$ and column weight $w_{\rm c}(m)$ expressed as functions of $m$, where no row or column consists of only zeros. Then, we say that $\F$ is a family of \emph{$q(m)$--sparse} quantum codes if
    \begin{enumerate}
        \item $\frac{w_{\rm r}(m)}{n_m} \leq q(m)$, and
        \item $\frac{w_{\rm c}(m)}{\ell_m} \leq q(m)$.
    \end{enumerate}
    
    If $q(m) = \O(\alpha^m)$ for some constant $\alpha\in(0,1)$, we say that every code in $\F$ has \emph{exponentially decaying density}.
    
    If $q(m)n_m \leq t$ and $q(m)\ell_m \leq t$ for some positive constant $t$, we say that every code in $\F$ is a $t$--qLDPC (or simply qLDPC) code.
\end{definition}

This definition is dependent on the representation of the parity-check matrices $H_m$ for $m \geq 1$.
One such representation is the standard form of the parity-check matrix, e.g. obtained by the algorithm in Chapter 9.4 of Ref.~\cite{djordjevic2021quantum}, where the matrix is full-rank and $\ell_m = n_m - k_m$.
However, in general, the parity-check matrix may be instructed by the experimental set-up, so we keep our definition general by introducing an arbitrary $\ell_m \geq n_m - k_m$. 
We then require that matrices $H_m$ do not contain any rows or columns of $H_m$ consisting of zeros, so that the sparsity of the matrices is not increased artificially.

For qLDPC codes, conditions $q(m)\ell_m \leq t$ and $q(m)n_m \leq t$ for constant $t>0$ correspond to the statements that every parity check acts on a constant number of qubits and every qubit participates in a constant number of parity checks.
Note that a $t$--qLDPC family $\F$ of codes with $n_m$ and $\ell_m$ that scale exponentially like $\O(\beta^m)$, for some $\beta > 0$, is itself a $\O(t/\beta^m)$--sparse code.
Therefore, qLDPC codes can be thought as codes with exponentially decaying density.
In Section~\ref{sec:construction}, we illustrate a construction (Algorithm~\ref{alg:gb_construction}) that produces new qLDPC codes and we then see in Section~\ref{sec:threshold} that they exhibit improved numerical performance with increasing code length. In Section~\ref{sec:scalable}, we see that sparsity has to be sacrificed in order to achieve scalability as defined later in Definition~\ref{def:scalability}.
Using our construction, we obtain a scalable family of $\O(\alpha^m)$--sparse codes, where $\alpha = 2/3$.
This family of codes with exponentially decaying density is still sparser than many known codes, such as the Shor code which scales as $\O(1/m)$, as we shown in Appendix~\ref{app:sparsity_shor_code}.

\subsection{Generalized Bicycle Codes}

Here, we lay out a construction provided in Ref.~\cite{Panteleev2021degenerate} for the generalized bicycle (GB) quantum codes~\cite{Kovalev2013quantum,mackay2004sparse}.
This ansatz forms the basis of our construction, and in particular our scalable family of small qLDPC codes.

Consider two binary $\ell \times \ell$ matrices $A$ and $B$ which commute, i.e. $[A,B]=0$.
Then, setting $H_X \coloneqq (A|B)$ and $H_Z \coloneqq (B^T|A^T)$ leads to a valid parity-check matrix $H$ as $H_X H_Z^T = AB + BA = [A,B] = 0$.
The obtained code $[[n,k]]$ has code length $n=2\ell$.
The special case of the original bicycle codes~\cite{mackay2004sparse} is obtained when $B = A^T$.

Random matrices $A$ and $B$ satisfy the commutativity condition with low probability, but fortunately there is a construction method which ensures the condition.
Let $\FF_2^{\langle \ell \rangle} \coloneqq \FF_2[x]/(x^\ell - 1)$ denote the ring of polynomials with binary coefficients modulo $x^\ell - 1$.
These are the polynomials 
\begin{equation}
    p(x) = p_0 + p_1 x + \dots + p_{\ell-1} x^{\ell-1},
\end{equation}
with coefficients $p_i \in \FF_2$ for $i=0,\dots,\ell-1$ and maximum degree $\ell-1$.
The weight $\wt p(x)$ of a polynomial $p(x)$ is the number of its non-zero coefficients, and ${\rm gcd}(\cdot)$ denotes the greatest common polynomial divisor of its arguments.

The ring $\FF_2^{\langle \ell \rangle}$ is isomorphic to the ring of $\ell \times \ell$ binary circulant matrices, such that a polynomial $p(x)$ is isomorphic to binary circulant matrix $P$, denoted by 
\begin{equation}
    p(x) \sim P,
\end{equation}
where
\begin{equation}\label{eq:circulant}
    P = 
    \begin{pmatrix}
        p_0 & p_{\ell-1} & \dots & p_1 \\ 
        p_1 & p_0 & \dots & p_2 \\
        \dots & \dots & \dots & \dots \\
        p_{\ell-1} & p_{\ell-2} & \dots & p_0
    \end{pmatrix}.
\end{equation}
To make the isomorphism transparent, we can express the circulant matrix as
\begin{equation}
    P =  p_0 I + p_1 X + \dots + p_{\ell-1} X^{\ell-1},
\end{equation}
where $I$ is the $\ell \times \ell$ identity and $X \coloneqq \sum_{k=0}^{\ell-1}\ketbra{k+1}{k}$ is the generalized Pauli $X$ operator in $\ell$ dimensions and $k+1$ denotes addition modulo $\ell$.

For any two circulant matrices $A$ and $B$ represented by polynomials $a(x)$ and $b(x)$, respectively, their product $AB$ is represented by the polynomial $a(x)b(x) \pmod{x^\ell - 1}$.
Therefore, the commutativity of $A$ and $B$, hence also the commutativity condition on the parity-check matrix $H$ defined by $A$ and $B$, is a consequence of the commutativity of polynomial multiplication over the ring $\FF_2^{\langle \ell \rangle}$.

Using this construction to obtain a GB code $[[2\ell,k]]$ with parity-check matrix $H$, it can be proven~\cite{Panteleev2021degenerate} that the code dimension $k$ is related to algebraic properties of the specific polynomials $a(x), b(x)$ that define $H$.

\begin{proposition}[GB code dimension~\cite{Panteleev2021degenerate}]
\label{prop:k}
    The dimension of a $GB$ code $[[2\ell,k]]$ defined by $a(x), b(x) \in \FF_2^{\langle \ell \rangle}$ is given by
    \begin{equation}
        k = 2\hspace{2pt}\deg g(x),
    \end{equation}
    where $g(x) \coloneqq {\rm gcd}\big(a(x), b(x), x^{\ell}-1 \big)$.
\end{proposition}

If one wishes to increase $\ell$, there is no obvious way to relate the properties of the new, larger code to the corresponding properties of the original code.
This is because, in general, new polynomials $a(x), b(x)$ of higher degree need to be constructed.
For example, the code dimension depends on the greatest common divisor of the polynomials which is a complicated function that does not simply depend on their degree.
Similar complications arise when one tries to relate the parity-check matrices of the codes.
In Section~\ref{sec:scalable}, we therefore present a simple methodology for obtaining a scalable family of GB codes.

\section{Algebraic extension of GB Codes}
\label{sec:construction}

In this section, we introduce a methodical extension for GB codes, that produces a family of qLDPC codes from an initial GB code ansatz which we call the \emph{base code}. 
We call the new codes extended GB codes, or simply \emph{extended codes}, as their generating polynomials are extensions of the generating polynomials for the base code.
The GB codes within this family are therefore related in the sense that one can obtain the full parity-check matrix of one code from smaller codes.

\subsection{Construction methodology}
\label{sec:algorithm}

The core idea of our construction lies in defining a sequence of algebraic transformations on the generating polynomials of the base code in order to obtain extended codes with related parity-check matrices. 
The algebraic transformations we allow convert the generating polynomials $a(x), b(x) \in \FF_2^{\langle \ell \rangle}$ of the base code into new polynomials $a'(x) = p(x)a(x) \in \FF_2^{\langle \ell' \rangle}$ and $b'(x) = p(x)b(x) \in \FF_2^{\langle \ell' \rangle}$ for some choice of $p(x)$ and for some $\ell'$ which is chosen as a positive integer larger than $\ell$.
We are thus sequentially enlarging the code lengths and the size of the parity-check matrices for the extended codes.
The new parameters $\ell',a'(x), b'(x)$ are described by $\ell, a(x), b(x)$ on their respective rings, allowing for a guarantee that the dimension of all codes in the family are lower bounded by the dimension of the base code. 
We describe the full method that constructs families of extended GB codes in Algorithm~\ref{alg:gb_construction}.

\begin{algorithm}[H]
\caption{Extended GB code construction}
\label{alg:gb_construction}
\begin{algorithmic}[1]
\Statex \textbf{Inputs:} A base code $\Q$ with generating polynomials $a(x), b(x) \in \FF_2^{\langle \ell \rangle}$;
The number $M$ of codes;
\Statex An increasing integer sequence $(\kappa_1, \dots, \kappa_M)$ with $\kappa_1 = 1$;
\Statex A sequence of polynomials $\big(p^{(1)}(x), \dots, p^{(M)}(x) \big)$, where $p^{(1)}(x) = 1$ and $p^{(m)}(x) \in \FF_2^{\langle (\kappa_m-1)\ell+1 \rangle}$ for $m=2, \dots, M$.
\vspace{5pt}
\State Initialize $\F \gets \{\}$.
\For{$m=1,\dots,M$} 
    \State Set $a^{(m)}(x) = p^{(m)}(x) a(x) \in \FF_2^{\langle \kappa_m \ell \rangle}$ and 
    \Statex \hspace{17pt} $b^{(m)}(x) = p^{(m)}(x) b(x) \in \FF_2^{\langle \kappa_m \ell \rangle}$.
    \vspace{1pt}
    \State Construct $(\kappa_m\ell \times \kappa_m\ell)$ circulant matrices 
    \Statex \hspace{17pt} $A_m \sim a^{(m)}(x)$ and $B_m \sim b^{(m)}(x)$.
    \vspace{1pt}
    \State Construct code $\Q_m$ via the parity-check matrix 
    \Statex \hspace{17pt} $H_m \coloneqq \begin{pmatrix} (A_m|B_m) & 0 \\ 0 & (B_m^T|A_m^T) \end{pmatrix}$.
    \vspace{1pt}
    \State Add $\Q_m$ to $\F$.
\EndFor
\vspace{5pt}
\Statex \textbf{Output:} A family of codes $\F \coloneqq \{ \Q_1, \dots, \Q_M \}$.
\end{algorithmic}
\end{algorithm}

The input requirements $\kappa_1 = 1$ and $p^{(1)}(x) = 1$ ensure that the base code $\Q$ is added in $\F$ as the first and smallest code $\Q_1$.
Moreover, each code $\Q_m \in \F$ is a $[[n_m, k_m]]$ code with code length that is increasing in $m$, as $n_m = 2\ell_m = 2\kappa_m\ell$.
Finally, $k_m = 2\hspace{1pt}\deg g^{(m)}(x)$, where $g^{(m)}(x) \coloneqq \gcd(a^{(m)}(x), b^{(m)}(x), x^{\ell_m} - 1)$ and $k_m \geq k_1$ as proved in Theorem~\ref{thm:guarantee}.

Each parity check of $\Q_m$ that appears as a row in the full parity-check matrix $H_m$ has weight equal to the sum of Hamming weights of $A_m$ and $B_m$, or equivalently
\begin{equation}\label{eq:alg1_wr}
    w_{\rm r}(m) = \wt a^{(m)}(x) + \wt b^{(m)}(x),
\end{equation}
and each qubit appears in 
\begin{equation}\label{eq:alg1_wc}
    w_{\rm c}(m) = \max\{\wt a^{(m)}(x), \wt b^{(m)}(x)\},
\end{equation}
parity checks, where $w_{\rm c}(m) < w_{\rm r}(m)$ always.

One can consider the asymptotic limit $M\rightarrow\infty$ to study the sparsity of the parity-check matrices.
We can immediately see from \eqref{eq:alg1_wr} and \eqref{eq:alg1_wc} that Algorithm~\ref{alg:gb_construction} produces a family of qLDPC codes if and only if the polynomials $p^{(m)}(x)$ do not increase the weight of the generating polynomials, i.e~$w_r(m) \leq q$ for all $m$ and some constant positive integer $q$.
This can be achieved by careful choice of $a(x), b(x)$ or of $p^{(m)}(x)$.

A simple way to achieve the LDPC property for arbitrary $a(x), b(x)$ is by setting $p^{(m)}(x) = 1$ for all $m$.
We then obtain a family of codes $\F_{\rm qLDPC} = \{\Q_1, \Q_2, \dots \}$, where $\Q_m$ is defined by polynomials $a(x), b(x) \in \FF_2^{\langle \kappa_m\ell \rangle}$ for $m = 1, 2, \dots$, which belong to rings of increasing dimension.
Therefore, the circulant matrices $A_m, B_m$ are all generated by polynomials $a(x), b(x)$, but they are different as $m$ increases.
In particular, the sequence $(\kappa_m)_{m=1,2,\dots}$ specifies how many additional zero elements $A_m$ and $B_m$ contain compared to $A_{m-1}$ and $B_{m-1}$, respectively.
We can see that $w_{\rm r}(m) = \wt a(x) + \wt b(x)$ for all $m = 1, 2, \dots$, which is independent of $m$, so $\F_{\rm qLDPC}$ forms indeed a family of qLDPC codes.

\subsection{Dimension and distance guarantees}

We now proceed by deriving a guarantee on the dimension of any family $\F$ that can be obtained by Algorithm~\ref{alg:gb_construction}.

\begin{theorem}\label{thm:guarantee}
    Suppose $\F$ is a code family obtained by Algorithm~\ref{alg:gb_construction}, consisting of $[[n_m, k_m, d_m]]$ codes $\Q_m \in \F$, where $m = 1, \dots, M$. Then, each code $\Q_m \in \F$ satisfies $k_m \geq k_1$.
\end{theorem}
\begin{proof}
    Let $a(x), b(x)$ be the polynomials that define a base code $\Q_1$ according to Algorithm~\ref{alg:gb_construction}.
    For any $[[2\kappa_m\ell, k_m]]$ code $\Q_m \in \F$, the code dimension is given by
    \begin{align}
        k_m &= 2\hspace{2pt}\deg \gcd\left(a^{(m)}(x), b^{(m)}(x), x^{\kappa_m\ell} - 1 \right) \nonumber\\
        &= 2\hspace{2pt}\deg \gcd\left( \gcd(a^{(m)}(x), b^{(m)}(x)), x^{\kappa_m\ell} - 1 \right) \nonumber\\
        &= 2\hspace{2pt}\deg \gcd\left( p^{(m)}(x)\phi(x), x^{\kappa_m\ell} - 1 \right),
    \end{align}where we have set $\phi(x) \coloneqq \gcd(a(x), b(x))$.
    All equalities follow from our definitions in Algorithm~\ref{alg:gb_construction}.
    We now note that the term $x^{\kappa_m\ell} - 1$ factors in the following way,
    \begin{align}
        x^{\kappa_m\ell} - 1 = (x^{\ell} - 1)\sum_{i=0}^{\kappa_m-1} x^{i\ell},
    \end{align}
    where all polynomials are defined over the ring $\FF_2^{\langle \kappa_m\ell \rangle}.$
    Finally, we recall the fact that for any polynomials $p(x), s(x), q(x) \in \FF_n^{\langle \kappa_m\ell \rangle}$, such that $\deg s(x) q(x) < \kappa_m\ell$, 
    \begin{align}
        \deg \gcd\big(p(x), q(x)s(x) \big) \geq \deg \gcd\big( p(x), q(x) \big),
    \end{align}which is applicable in this scenario as we state in Proposition~\ref{thm:gcd_inequality} of Appendix~\ref{app:prelim}.
    We can use this fact twice to bound the degree of the greatest common divisor,
    \begin{align}
        k_m &= 2\hspace{2pt}\deg \gcd\left( p^{(m)}(x)\phi(x), (x^{\ell} - 1)\sum_{i=0}^{\kappa_m-1} x^{i\ell} \right) \nonumber\\
        &\geq 2\hspace{2pt}\deg \gcd\left( \phi(x), (x^{\ell} - 1)\sum_{i=0}^{\kappa_m-1} x^{i\ell} \right) \nonumber\\
        &\geq 2\hspace{2pt}\deg \gcd\big( \phi(x), x^{\ell} - 1 \big) \nonumber\\
        &= k_1,
    \end{align}
    concluding the proof.
\end{proof}

The following theorem describes an explicit formula for the dimension of the $m$-th parity-check matrix under an additional constraint on the generating polynomials.

\begin{theorem}\label{thm:guarantee2}
    Suppose $\F$ is a code family obtained by Algorithm~\ref{alg:gb_construction}, consisting of $[[n_m, k_m, d_m]]$ codes $\Q_m \in \F$, where $m = 1, \dots, M$ and $p^{(m)}(x),\gcd(a(x), b(x))$ are coprime. Then, each code $\Q_m \in \F$ satisfies
    \begin{widetext}
    \begin{eqnarray}
        k_m = k_1 + 2(\deg \gcd(p^{(m)}(x),x^\ell - 1)+ \deg\gcd(p^{(m)}(x), \sum_{i=0}^{\kappa_m - 1}x^{il}) + \deg\gcd(a(x), b(x), \sum_{i=0}^{\kappa_m - 1}x^{il}) ).\nonumber
    \end{eqnarray}
    \end{widetext}
\end{theorem}
\begin{proof}
    Let $a(x), b(x)$ be the polynomials that define a base code $\Q_1$ according to Algorithm~\ref{alg:gb_construction}.
    Following the reasoning of the proof of the previous theorem, the code dimension of any code $\Q_m \in \F$ is $k_m = 2\hspace{2pt}\deg \gcd\left( p^{(m)}(x)\phi(x), x^{\kappa_m\ell} - 1 \right)$, where we have set $\phi(x) \coloneqq \gcd(a(x), b(x))$, and $x^{\kappa_m\ell} - 1 = (x^{\ell} - 1)\sum_{i=0}^{\kappa_m-1} x^{i\ell},$
    where all polynomials are defined over the ring $\FF_2^{\langle \kappa_m\ell \rangle}.$
    Finally, we recall the fact that for polynomials $p(x), s(x), q(x) \in \FF_n^{\langle \kappa_m\ell \rangle}$, where $q(x)$ and $s(x)$ are coprime, $\gcd(p(x), q(x)s(x)) = \gcd(p(x), q(x))\gcd(p(x), s(x))$ \cite{niederreiter2009finitefields}.
    Thus, we can conclude that
    \begin{widetext}
    \begin{eqnarray}
        k_m &=& 2\hspace{2pt}\deg \gcd\left( p^{(m)}(x)\phi(x), (x^{\ell} - 1)\sum_{i=0}^{\kappa_m-1} x^{i\ell} \right) \nonumber\\
        &=& 2\hspace{2pt}\deg\left[ \gcd\left(p^{(m)}(x)\phi(x), x^{\ell} - 1\right)\cdot
        \gcd\left( p^{(m)}(x)\phi(x), \sum_{i=0}^{\kappa_m-1} x^{i\ell} \right)\right]\nonumber\\
        &=& 2\hspace{2pt}\deg \gcd\left(p^{(m)}(x), x^{\ell} - 1\right)+ 2\hspace{2pt}\deg\gcd\left(\phi(x), x^{\ell} - 1\right)\nonumber\\
        &+& 2\hspace{2pt}\deg\gcd\left(p^{(m)}(x), \sum_{i=0}^{\kappa_m-1} x^{i\ell} \right)+ 2\hspace{2pt}\deg\gcd\left(\phi(x), \sum_{i=0}^{\kappa_m-1} x^{i\ell} \right)\nonumber\\
    \end{eqnarray}
    \end{widetext}where we have used in the first equality the fact that $(x^{\ell} - 1), \sum_{i=0}^{\kappa_m-1} x^{i\ell}$ are coprime, and the assumption that $p^{(m)}(x),\gcd(a(x), b(x))$ are coprime in the second equality. Since $2\deg\gcd (\phi(x),x^{\ell} - 1)=k_1,$ the result follows.
\end{proof}

The requirement that $\ell_m$ be an integer multiple of $\ell$ is crucial, otherwise there is no guarantee that the code dimension $k_m$ is lower bounded by $k_1$.
In fact, $k_m$ can otherwise be 0, as we show in Appendix~\ref{app:kguarantee} by an explicit example.

\section{Threshold analysis for extended codes}
\label{sec:numerics}

In this section, we present numerically obtained thresholds for small qLDPC codes obtained by the construction of Algorithm~\ref{alg:gb_construction}.

\subsection{Decoding via belief propagation}

In order to obtain the threshold values for a family of codes $\F$, we run several noise simulations so that for each run an error string $\bme$ occurs, leading to a syndrome
\begin{equation}\label{eq:syndrome}
    \bms = H \cdot \bme.
\end{equation}
The goal of a decoder is to obtain the most likely error string $\bme$ that satisfies this syndrome \eqref{eq:syndrome}. 
Decoding for our simulations was performed using the $BP+OSD$ decoder~\cite{roffe_decoding_2020} which we implemented via the \texttt{bposd} Python package~\cite{Roffe_LDPC_Python_tools_2022}.

\emph{Belief Propagation (BP)}~\cite{gallager1963low,Kschischang2001,mackay1995good} is the most frequently used decoder for classical LDPC codes~\cite{gallager1963low}.
It iteratively updates the probability distribution $P(\bme|\bms)$ of individual bits in the codeword.
In the classical setting, BP aims to find the exact error string $\bme$ that satisfies the syndrome equation (\eqref{eq:syndrome}), by updating the error string $\bme \mapsto \bme'$ in a sequence of ``beliefs'', with the $i$-th bit given by
\begin{equation}
    e_i \rightarrow
    \begin{cases} 
        1 &\text{ if }\hspace{5pt} p(e_i) \geq 0.5 \,, \\ 
        0 &\text{ if }\hspace{5pt} p(e_i) < 0.5 \,,
    \end{cases}
\end{equation}
where the conditional probability of an error occurring on the $i$-th bit given the syndrome $\bms$ equals to $p(e_i) = \sum_{e_j \neq e_i} P(e_1,\ldots,e_{i-1},1,e_{i+1},\ldots,e_n|\bms)$.

In the quantum setting, there may be multiple minimum weight estimates of the error for a given syndrome, due to the stabilizer encoding, a phenomenon known as \emph{quantum degeneracy}.
It is then sufficient to find a recovery operation $\bmr$ that corrects the error $\bme$ up to a stabilizer, so that $\bmr + \bme \in {\rm rowspace}(H)$.
However, BP fails to account for quantum degeneracy as it assigns high probabilities to all minimum weight error estimates, thus failing to converge.

\emph{Ordered Statistics Decoding (OSD)}~\cite{fossorier1995soft,Panteleev2021degenerate} is a post-processing algorithm supplementing BP with success for various families of qLDPC codes~\cite{Panteleev2021degenerate,roffe_decoding_2020}.
Even though
belief propagation may fail to provide a hard decision
on the error estimate, it always outputs marginal probabilities for the error estimates that satisfy the syndrome
equation. The idea behind OSD is to then use these probabilities (soft decisions) in order to select syndrome bits
with a high probability of having flipped. An OSD algorithm of order 0, denoted by OSD-–0, selects the columns $h=(\bm{h}_i)_{i=1,\ldots,\text{rank}(H)}$ of the parity-check matrix that correspond to the highest probabilities and obtains an estimate $h^{-1}\bm{s}$. In general, the OSD algorithm can be of order $w$, indicating a greedy search over increasingly more columns, improving decoding performance. 
Full details on the $BP+OSD$ protocol and its implementation can be found in Refs.~\cite{Panteleev2021degenerate,roffe_decoding_2020}.


In the simulations we present in this section, we consider a phenomenological model consisting of unbiased depolarizing noise, i.e. a single-qubit Pauli flip $\rho \mapsto P\rho P$ occurs with constant probability $p$ independently of the choice of Pauli $P \in \{X,Y,Z\}$.
We use the `min-sum' variant of BP~\cite{emran2014simplified} (as described in detail in Appendix C of Ref~\cite{roffe_decoding_2020}) with scaling factor $0.625$ and iteration depth 40.
The OSD order is set to $w = \ell$, where $\ell$ is the dimension of the polynomial ring of the base code.
We run $5 \cdot 10^4$ simulations, where each simulation is terminated at a precision of $10^{-3}$, a value much lower than any obtained threshold values.

\subsection{Choosing a base code}
\label{sec:basecode}

We would like to generate various families $\F_{\rm qLDPC}$ of qLDPC codes, according to the construction laid out in \secref{sec:algorithm}.
Our main motivation is to discover families of qLDPC codes which display improved performance with increasing code length, while keeping the code lengths as small as possible.

\begin{figure}[t]
    \centering 
    \includegraphics[width=1\columnwidth]{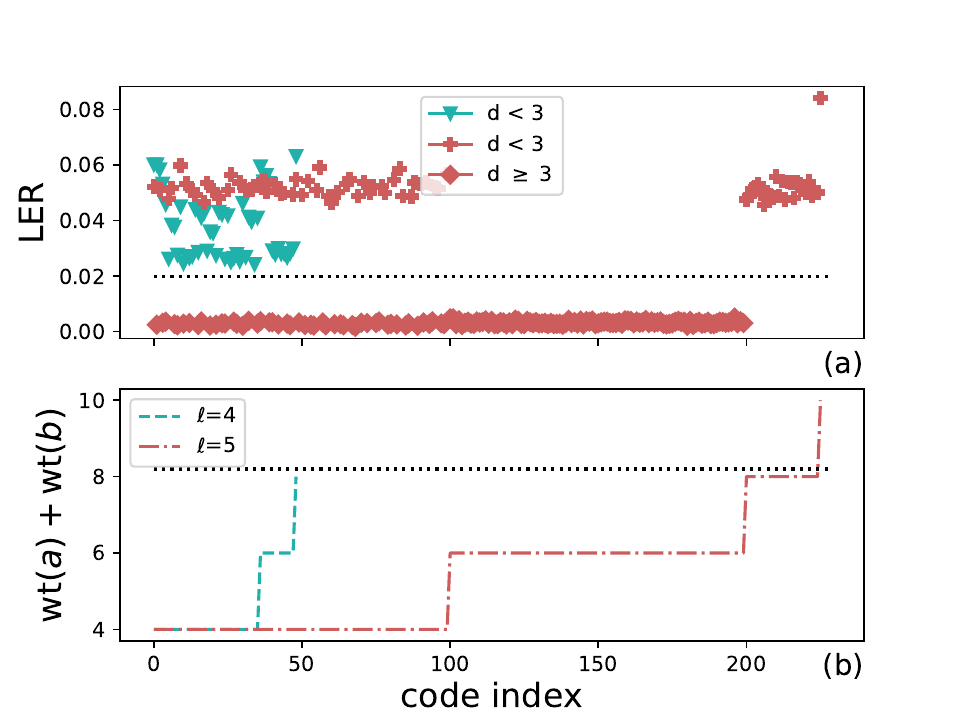}
    \caption{ {\bf Small GB base codes.}
    \textbf{(a)} Logical error rate (LER) of GB codes at ${\rm PER} = 0.010$.
    We include all GB codes with non-zero dimension at $\ell = 4$ and $\ell = 5$ (49 and 226 in total respectively).
    Codes are marked with triangles which are pointing down if the distance is less than 3 and pointing up if the distance is at least 3.
    It is clear that we need $\ell \geq 5$ to achieve distance 3, and only codes with distance 3 produce a low LER (below the dashed line at 0.020).
    \textbf{(b)} Maximum number of qubits participating in each parity check, equal to the sum of the polynomial weights $q = \wt(a) +\wt(b)$.
    The codes are ordered in increasing $q$, and we cap the weight of the codes we consider at 8 (dashed line).
    As $\ell$ increases beyond 5, this restricts the fraction of codes we consider more.
    }
    \label{fig:distance_three}
\end{figure}

In order to construct a particular family $\F_{\rm qLDPC}$, we first need to identify a suitable GB base code $\Q_1$.
To this end, we search for polynomials $a(x),b(x)$ over the ring $\FF_2^{\langle \ell \rangle}$.
By construction, the generated code is guaranteed to satisfy the commutativity condition in \eqref{eq:commutativity}.
However, we need to ensure that polynomials $a(x),b(x)$ satisfy the additional property that $\deg g(x) > 0$, where $g(x) \coloneqq {\rm gcd}\big(a(x), b(x), x^{\ell}-1 \big)$, otherwise the code has dimension 0 according to Proposition~\ref{prop:k}.

For our numerical investigations in \secref{sec:threshold}, we perform an exhaustive search over all possible pairs of polynomials $a(x), b(x) \in \FF_2^{\langle \ell \rangle}$ for $\ell$ up to 10.
We discard codes for which $\deg g(x) = 0$ and we find that $\ell = 5$ is the smallest value that produces error-correcting codes with distance $d \geq 3$.
To illustrate this, we present in \figref{fig:distance_three} all codes that satisfy $\deg g(x) > 0$ for $\ell=4$ and $5$.
In particular, we plot the logical error rate (LER) at physical error rate (PER) equal to $0.010$.
Codes with distance less that 3 are marked with triangles pointing down for $\ell=4$ and crosses for $\ell=5$, while codes with distance 3 are marked with diamonds.
We can see that all codes of distance 3 correspond to $\ell=5$ and display a low LER, below the dashed threshold line drawn at ${\rm LER} = 0.020$, while all codes with distance less than 3 (with $\ell = 4$ or $5$) produce very high LER, creating a considerable gap from the well-behaving distance 3 codes.
Therefore, the smallest GB codes one can consider have code length 10.

Since calculation of distance soon becomes costly, we extrapolate this dashed threshold line to higher $\ell$. 
Hence, given a particular $\ell$, we consider all codes for which $\deg g(x) > 0$ and ${\rm LER} < 0.020$ at ${\rm PER} = 0.010$.
Due to the interest in realizing such codes experimentally, we further limit the sum of polynomial weights $t = \wt a(x) + \wt b(x)$ by 8, so we obtain $t$--qLDPC codes with $t \leq 8$.
This is more restrictive at values of $\ell$ higher than 5, while still allowing for a very large fraction of available codes.
Out of those, we get many examples of codes with $\ell = 5, \dots, 10$ that display improved performance as the code length increases according to the qLDPC construction of Algorithm~\ref{alg:gb_construction}.
We discuss one example of the smallest obtainable family in the following section, and we give examples for all $\ell \leq 10$ in Appendix~\ref{app:numerics}.

\subsection{Numerical thresholds}
\label{sec:threshold}

We would like to construct the smallest family $\F_{\rm qLDPC}$ allowed by Algorithm~\ref{alg:gb_construction}.
Firstly, we pick $\ell = 5$, the smallest value that guarantees code distance $d \geq 3$.
We then perform a search for a base code as described in \secref{sec:basecode}, which finds the $[[10, 2, 3]]$ GB code defined by polynomials $a(x) = 1 + x^4$ and $b(x) = 1 + x + x^2 + x^4$.
The parity-check matrix of this code has weight $w_{\rm r} = \wt a(x) + \wt b(x) = 6$.
We then proceed to extend this code into new codes such that the code lengths increase in smallest steps.
This is achieved by choosing the integer arithmetic sequence $(\kappa_1, \kappa_2, \dots)$ with $\kappa_m = m$ for $m \geq 1$, leading to code lengths $10, 20, 30, \dots$ for codes in $\F_{\rm qLDPC}$.
Finally, we choose $p^{(m)}(x) = 1$ to ensure a constant $w_{\rm r}(m) = 6$ for all $m$, creating the 6--qLDPC family $\F_{\rm qLDPC}$.

\begin{figure}[t]
    \centering 
    \includegraphics[width=1\columnwidth]{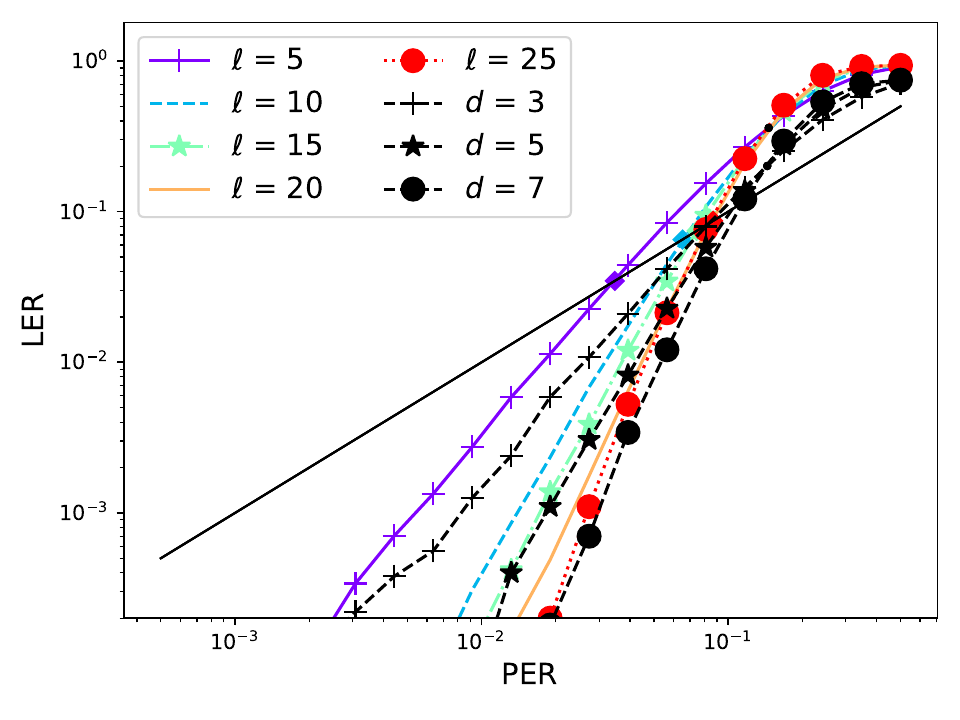}
    \caption{ {\bf Small qLDPC code family obtained from Algorithm~\ref{alg:gb_construction}.} 
    We plot the performance of a family $\F_{\rm qLDPC}$ of small qLDPC codes.
    The $[[10, 2, 3]]$ base code has the smallest code length that guarantees distance $d \geq 3$, and it is constructed by polynomials $a(x) = 1 + x^4$ and $b(x) = 1 + x + x^2 + x^4$.
    The construction of the remaining codes in $\F_{\rm qLDPC}$ follows Algorithm~\ref{alg:gb_construction} with $\kappa_m = m$ and $p^{(m)}(x) = 1$, leading to qLDPC codes with $\ell_m = 5m$ for $m = 1, \dots, 5$.
    We compare the performance of the first five codes in $\F_{\rm qLDPC}$ with the surface code of distances $d = 3,5,7$.
    The performance of codes in $\F_{\rm qLDPC}$ improves with increasing code length and the threshold point is around ${\rm PER} = 0.145$, similar to the surface code threshold that we plot.
    }
    \label{fig:threshold}
\end{figure}

We plot the five smallest codes of $\F_{\rm qLDPC}$ in \figref{fig:threshold}.
We observe a threshold value at around $PER = 0.145$, which is similar to the surface code threshold plotted for comparison.
The qLDPC codes of $\ell = 10$ and $\ell = 20$ are shown to outperform the surface codes of distance 3 and 5 respectively.
The qLDPC code of $\ell = 25$ requires 50 data qubits and 48 syndrome qubits for a total of 98 qubits.
In comparison the surface code of distance 7 requires a total of 97 qubits.
These two codes appear to exhibit similar performance.
Note that the qLDPC family encodes 2 logical qubits in contrast to the surface code which always encodes 1 logical qubit.

Following the same construction, i.e. using a GB base code, $\kappa_m = m$ and $p^{(m)}(x) = 1$, we find qLDPC code families with performance that improves as the code length increases at every $\ell = 5, \dots, 10$.
We present examples of these codes in Appendix~\ref{app:numerics}.

\section{Scalable quantum codes from LDPC codes}
\label{sec:scalable}

Given limitations of superconducting quantum computing platforms \cite{Bravyi2022}, it is important to take scalability into account when designing a quantum code.
Assuming that a quantum code is already embedded in a chip, it can then be easily replicated to create multiple identical quantum memories.
Once the computational capabilities need to scale, we might require a larger code to encode logical information.
In such a scenario, it is desirable to scale  current chips without imposing completely different connectivity,
making, therefore, procedural the development of chips related to a family of qLDPC codes. In this section, we consider this notion of scalability for a family of extended GB codes.

This is a strict notion of scalability as a larger code does not lead to any alterations in the connectivity between the original chip components (data and stabilizer qubits).
We can view this notion of scalability as the property that the parity-check matrix of a smaller code is embedded in the parity-check matrix of a larger code up to relabelling of qubits.
We provide a definition below.
\begin{definition}\label{def:scalability}
    Let $\F \coloneqq \{\Q_1, \dots, \Q_M \}$ be a family of $[[n_m, k_m]]$ quantum CSS codes for $m=1,\dots,M$.
    Then, $\F$ is \emph{scalable} if
    \begin{enumerate}
        \item for any two codes $\Q_{m'}, \Q_{m} \in \F$ with parity-check matrices given by $(2\ell' \times 2n')$ matrix $H'$ and $(2\ell \times 2n)$ matrix $H$, respectively, we have that $H$ is embedded in $H'$ whenever $n' > n$; i.e., there exists a sequence of column and row swaps such that
        \begin{equation}
            (H_{X}')_{ij} = (H_{X})_{ij} \quad\text{ if } (H_{X})_{ij} = 1,
        \end{equation}and
        \begin{equation}
            (H_{Z}')_{ij} = (H_{Z})_{ij} \quad\text{ if } (H_{Z})_{ij} = 1,
        \end{equation}
        for all $1 \leq i \leq \ell$ and $1 \leq j \leq n$,
        
        \item there exists a constant $p_{th}$ such that for all PER $p < p_{th}$, we have that the robustness of the code $\Q_{m'}$ is improved when compared with the code $\Q_{m}$, for all $m' > m$. In practice, we quantify robustness by the LER of the code.
    \end{enumerate}
\end{definition}

This definition is intuitively satisfied by topological codes, such as the surface code.
As the code length grows, the data qubits remain connected to the same syndrome qubits, up to relabelling, with the boundary data and syndrome qubits requiring new connections to newly introduced syndrome and data qubits, respectively.

If an qLDPC code has no boundary qubits, in the sense that each qubit participates in $q$ parity checks and each parity check involves $q$ qubits, then it is not possible to scale the code according to Definition~\ref{def:scalability}.
This is expected since requiring that the bigger parity-check matrix $H_{m'}$ embeds the smaller one $H_{m}$ leads to the form
\begin{equation}
    H_{m'} = \begin{pmatrix} H_m & \bmO \\ \bmO & H' \end{pmatrix},
\end{equation}
for some matrix $H'$.
The zero off-diagonal blocks $\bmO$ are necessary to ensure that the new code $\Q_{m'}$ is qLDPC, hence there is no connection between the original code $\Q_m$ and the new code $\Q_{m'}$. Therefore, we have two independent codes instead of one code with higher code length.
Note that the qLDPC construction of Algorithm~\ref{alg:gb_construction} for example demonstrated in Section~\ref{sec:threshold} falls under this scenario.

We are therefore forced to sacrifice sparsity of the parity-check matrix to ensure scalability when one considers extended GB codes via the construction of Algorithm~\ref{alg:gb_construction}.
A natural question to ask then is:
\begin{center}
    \emph{What is the minimal amount of connectivity that one needs to add to the original chip components to guarantee scalability according to Definition~\ref{def:scalability}?}
\end{center}
We are able to answer this question using Algorithm~\ref{alg:gb_construction} by constructing a scalable family of $\O((2/3)^m)$--sparse extended GB code, i.e.~a family of codes with a parity-check matrix density that decays exponentially with increasing code length according to Definition~\ref{def:sparsity}.

Given a matrix C, we define the lower triangular matrix $L(C)$ as the matrix obtained from $C$ by setting all elements above  the main diagonal $C_{ij}$ with $j > i$ to 0.
Similarly, we define the upper triangular matrix $U(C)$ as the matrix obtained from $C$ by setting all elements below and including the main diagonal $C_{ij}$ with $j \leq i$ to 0.
Using these definitions, we describe in Theorem~\ref{thm:scalable} our new family of extended GB codes that satisfy the scalability condition of Definition~\ref{def:scalability}.

\begin{restatable}{theorem}{mainclaim}\label{thm:scalable}
    Given a code $\Q_1$ defined by  $\ell \times \ell$ binary circulant matrices $A_1$ and $B_1$ via the parity-check matrix
    \begin{equation}
        H_1 \coloneqq \begin{pmatrix} (A_1|B_1) & 0 \\ 0 & (B_1^T|A_1^T) \end{pmatrix},
    \end{equation}
    one can construct a family of $[[n_m, k_m,d_m]]$ quantum codes $\F \coloneqq \{\Q_1, \dots, \Q_M \}$, such that:
    \begin{enumerate}
        \item\label{en:parameters} The codes in $\F$ satisfy
        \begin{equation}
            n_{m+1} = 3n_m \text{ and } k_{m+1} \geq k_m,
        \end{equation}
        for all $m = 1,\dots,M-1$;
        \item\label{en:Hm} $\Q_{m+1}$ is obtained from $\Q_m$ for $m = 1,\dots,M-1$ by setting $A_{m+1} = F(A_m)$ and $B_{m+1} = F(B_m)$, where function $F$ acts on matrices as
        \begin{equation}\label{eq:pc_inductive}
            F(C) \coloneqq \begin{pmatrix} L(C) & U(C) & C \\ C & L(C) & U(C) \\ U(C) & C & L(C) \end{pmatrix}.
        \end{equation}
        \item\label{en:sparsity} $\F$ is a family of $q(m)$--sparse code according to Definition~\ref{def:sparsity} with
        \begin{equation}\label{eq:scalable_qm}
            q(m+1) = \left(\frac{2}{3}\right)^{m} \frac{\max\{\wt a(x), \wt b(x)\}}{\ell},
        \end{equation}
        for $m = 1,\dots,M$,
        where polynomials $a(x), b(x)$ are isomorphic to $A_1, B_1$ respectively;
        \item\label{en:scalability} $\F$ is scalable according to Definition~\ref{def:scalability}.
    \end{enumerate}
\end{restatable}

We prove Theorem~\ref{thm:scalable} in Appendix~\ref{app:scalability}.

To illustrate the connectivity of the extended code, consider a representation of a chip architecture of the extended code shown in Fig.~\ref{fig:chip_connectivity}. The chip architecture for the $m$-th family member comprises three identical chips depicted as dashed rectangles. Each chip consists of data and syndrome qubits (black dots) which have the $(m-1)$-th family member connectivity, which is schematically indicated by a cloudy hatching. 
The $(m-1)$-th family member connectivity is the required connectivity to implement the $(m-1)$-th code from that family. Further, there is connectivity between the qubits of the first and the second chips, between the qubits of the second and the third chips, and between the qubits of the first and the third chips such that the connectivity between all qubits is according to the required connectivity of the extended code given by Eq.~(\ref{eq:pc_inductive}). It is, therefore, evident that the code construction shown in Theorem~\ref{thm:scalable} are particularly suitable for implementation on state-of-the-art quantum computing architectures wherein chips with a relatively small number of qubits are replicated and connected to allow for an implementation of QEC codes with larger code length, and potentially larger dimension and distance.

\begin{figure}[t]
    \centering 
    \includegraphics[width=1\columnwidth]{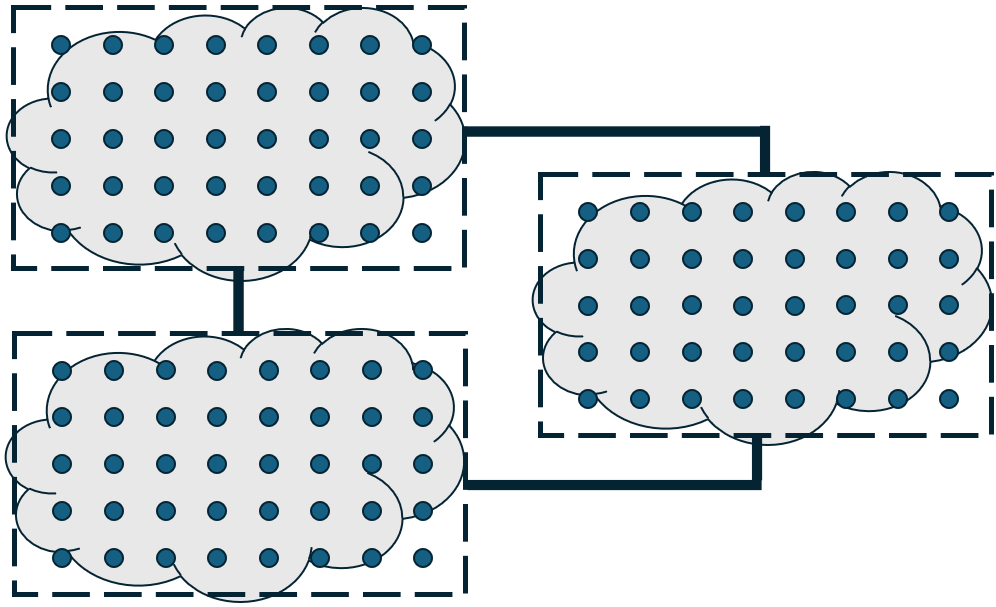}
    \caption{ {\bf Chip architecture of the extended code of Theorem~\ref{thm:scalable}.} The chip architecture of the $(m+1)$-th family member is composed of three identical chips describing the connectivity of the $m$-th family member.
    Each chip consists of data and syndrome qubits represented by black dots. To avoid imposing a particular connectivity on the $m$-th family member, we have its schematic indicated by a cloudy hatching. Notice that there is connectivity between the qubits of the first and the second chips, between the qubits of the second and the third chips, and between the qubits of the first and the third chips such that the connectivity between all qubits is according to the required connectivity of the extended code.}
    \label{fig:chip_connectivity}
\end{figure}

The following theorem describes an additional scalable construction of qLDPC codes given in this paper. This construction focuses on preserving the symmetry present on the parity-check matrix of the base code when we create the extended code.

\begin{restatable}{theorem}{secondmainclaim}\label{thm:scalable2}
    Given a code $\Q_1$ defined by  $\ell \times \ell$ binary circulant matrices $A_1$ and $B_1$ via the parity-check matrix
    \begin{equation}
        H_1 \coloneqq \begin{pmatrix} (A_1|B_1) & 0 \\ 0 & (B_1^T|A_1^T) \end{pmatrix},
    \end{equation}
    one can construct a family of $[[n_m, k_m,d_m]]$ quantum codes $\F \coloneqq \{\Q_1, \dots, \Q_M \}$, such that:
    \begin{enumerate}
        \item\label{en:parameters2} The codes in $\F$ satisfy
        \begin{equation}
            n_{m+1} = n + 2m\cdot r \text{ and } k_{m+1} \geq k_m,
        \end{equation}for $m = 1,\dots,M-1$ and $r$ integer;
        \item\label{en:Hm} $\Q_{m+1}$ is obtained from $\Q_m$ for $m = 1,\dots,M-1$ by setting $A_{m+1} = F(A_m)$ and $B_{m+1} = F(B_m)$, where function $F$ acts on matrices as
        \begin{equation}\label{eq:pc_inductive2}
            F(C) \coloneqq 
            \begin{pmatrix}
                c_0        & c_{\ell-1} & \dots & c_1 \\ 
                c_1        & c_0        & \dots & c_2 \\
                \dots      & \dots      & \dots & \dots \\
                c_{j-2}    & \dots      & \dots & c_{j-1} \\
                c_{j-1}    & c_{j-2}    & \dots & 0 \\
                0          & c_{j-1}    & \dots & 0 \\
                \dots      & \dots      & \dots & \dots \\
                0          & 0          & \dots & c_{j} \\
                c_{j}      & 0          & \dots & c_{j+1} \\
                c_{j+1}    & c_{j}      & \dots & c_{j+2} \\
                \dots      & \dots      & \dots & \dots \\
                c_{\ell-1} & c_{\ell-2} & \dots & c_0
            \end{pmatrix},
        \end{equation}where $c(x)$ is isomorphic to the matrix $C$ and $0<j<\ell-1$.
        \item\label{en:sparsity} $\F$ is a family of $t$-qLDPC codes according to Definition~\ref{def:sparsity} with
        $t = \wt a(x) + \wt b(x)$, where polynomials $a(x), b(x)$ are isomorphic to $A_1, B_1,$ respectively;
    \end{enumerate}
\end{restatable}

We prove Theorem~\ref{thm:scalable2} in Appendix~\ref{app:scalability2}.

We have proposed two code constructions with interesting properties. On the first hand, we have shown in Theorems~\ref{thm:scalable} how to utilize a GB code ansatz to build a family of codes which is scalable according to Definition~\ref{def:scalability}. Scalability is obtained in Theorem~\ref{thm:scalable} at the expense of weakening the LDPC property to the property that parity-check matrices have exponentially decaying density. 
Observing the inductive expression in \eqref{eq:pc_inductive} that describes how the $(m+1)$-th parity-check matrix $H_{m+1}$ is obtained from the $m$-th parity-check matrix $H_m$, we see that each circulant matrix $A_m,$ and $B_m$ is embedded three times within $A_{m+1},$ and $B_{m+1}$, respectively, each time involving a separate set of qubits and stabilizers.
The remaining six blocks that make up $A_{m+1},$ and $B_{m+1}$ are lower and upper triangular matrices obtained from $A_m,$ and $B_m$, so they contain a strict subset of the connections between syndrome and data qubits that appear in $A_m,$ and $B_m$. Therefore, $H_{m+1}$ locally shares the structure of $H_m$. 
On the second hand, we have given the development of families of qLDPC codes in Theorem~\ref{thm:scalable2} by appending a zero-vector to the first column of the parity-check matrix $H_{m}$ and imposing the resultant parity-check matrix to be a circulant matrix. Since column and row weights are maintained, the family of codes from Theorem~\ref{thm:scalable2} is indeed qLDPC. Unfortunately, the embedding property described above for the family of codes in Theorem~\ref{thm:scalable} cannot be stated for the family in Theorem~\ref{thm:scalable2}. However, as can be seen from the inductive expression in \eqref{eq:pc_inductive2} and some elementary matrix manipulation, the $m$-th parity-check matrix $H_{m}$ obtained from the GB parity-check matrix $H_1$ has circulant matrices with only few additional long-range connections.

\section{Conclusion}
\label{sec:summary}

We have introduced quantum code construction methods of small CSS codes. 
Our constructions notably differ from previous attempts presented in literature in that they are neither topological nor based on the manipulation of the Tanner graph of quantum codes.
Instead, they are based on directly manipulating parity-check matrices of GB codes, which is more akin to the construction of quantum product codes based on the parity-check matrices of classical codes. As a consequence, we describe two new techniques to construct families of quantum LDPC codes from extended GB codes that satisfy the scalability condition and one technique that produces sparse codes. In particular, the latter technique can be used to generate a family of codes that can be implemented on chips which embed the connectivity graph of smaller members of the same family. This property is particularly suitable for implementation on state-of-the-art quantum computing architectures wherein chips with a relatively small number of qubits are replicated, therefore making our code families scalable.

We have numerically demonstrated that our construction leads to quantum LDPC codes of performance which is numerically comparable to the surface code for minimum distances of up to seven. 
Here, we quantify performance in terms of the threshold value and the lambda factor, which is the error suppression obtained by increasing the code member index in the family. By using this methodology, we presented a family of qLDPC codes which exhibits a similar performance compared to the surface code family, but encodes two logical qubits instead of one. Furthermore, we noticed that the performance curve of a qLDPC family member has higher slope than the distance seven surface code, which means that for some particular physical error rate the performance of our qLDPC code is superior to the surface code.

On the experimental side, we know that non-planar connectivity is necessary to obtain good quantum codes in terms of error correction capabilities. Nevertheless, it is imperative to also take into account real-world device capabilities in order to realize early experimental demonstrations of small scalable quantum codes. 
Our work attempts to do precisely that by designing families of codes fitting the needs of real-world quantum devices, specifically those with connectivity graphs that are scalable or of a constant vertex degree. One can i.e. propose a chip design implementing the syndrome-extraction circuit for the code constructions that we presented in this paper by using using flip-chip technology, resonators or SWAP networks. We aim to determine which of these strategies gives optimum performance in future studies.

This paper also suggests other future lines of investigation. Firstly, our numerical performance is evaluated under a phenomenological noise model and we it would be of interest to study fault-tolerant syndrome-extraction circuits and analyze them under circuit-level noise. Since qLDPC codes traditionally have at least weight six stabilizers, an analysis under circuit-level noise is important to understand the cost of having higher stabilizer weights, larger circuit depths, as well as more hook errors per stabilizer. 
Secondly, an obvious direction arising from our work is exploring further properties of codes belonging to a given family. In particular, it could be relevant to create families of codes with minimum distance guarantees. In this way, one could reduce the search space of code families. Thirdly, it is crucial for the progress and practical implementation of quantum error-correcting codes and, ultimately, fault-tolerant quantum algorithms that we identify families of quantum error-correcting codes which are implementation friendly. We have made the first step in this direction in this paper. However, it would be desirable to pin down particular families of codes of interest, which fit additional geometric and implementation constraints, i.e., codes with low thickness and local non-planar connectivity graphs. 
Lastly, a more comprehensive theory describing how to change our construction algorithm to other code types than GB codes can pave the way for exploring further code constructions. A careful modification of our algorithm using hypergraph- or lifted product codes can produce code families which are simultaneously asymptotically good and implementation friendly.

\bibliographystyle{quantum}
\bibliography{ref}

\onecolumn
\appendix

\section{Supplementary analytical results}
\label{app:analytics}

Here we provide some technical analytical statements. In particular, we recall some known statements about matrices and polynomial rings over finite fields. Then, we prove the results presented in Theorems~\ref{thm:scalable}~and~\ref{thm:scalable2}. The sparsity of the generalized Shor code is discussed. Lastly, we give further examples of families of codes obtained from Algorithm~\ref{alg:gb_construction}.

\subsection{Preliminaries}
\label{app:prelim}

We first recall our definition of upper and lower triangular matrices.
\begin{definition}
    Given an $n \times n$ matrix $C$, we call the $n \times n$ matrix denoted by $L(C)$ \emph{lower triangular} if $L(C)_{ij} = C_{ij}$ for all $j \leq i$ and $L(C)_{ij} = 0$ for all $j > i$.
    
    Furthermore, we call the $n \times n$ matrix denoted by $U(C)$ \emph{upper triangular} if $U(C)_{ij} = C_{ij}$ for all $j > i$ and $U(C)_{ij} = 0$ for all $j \leq i$.
\end{definition}

We then state a useful standard result known in the literature~\cite{niederreiter2009finitefields}.

\begin{proposition}\label{thm:gcd_inequality}
    Let $\FF_n^{\langle \ell \rangle}$ be a polynomial ring for some positive integers $n,\ell$.
    Then, for any $p(x), s(x), q(x) \in \FF_n^{\langle \ell \rangle}$, such that $\deg s(x) q(x) < \ell$, 
    \begin{align}
        \deg \gcd\big(p(x), s(x)q(x) \big) \geq \deg \gcd\big( p(x), q(x) \big).
    \end{align}
\end{proposition}

This result is used in Theorem~\ref{thm:guarantee} and we will use it to show $k_{m+1} \geq k_m$ in the proof of Theorem~\ref{thm:scalable} in Appendix~\ref{app:scalability}.

\subsection{Dimension guarantee}
\label{app:kguarantee}

Here, we construct an explicit example to show that, for a family of codes $\F = \{\Q_1, \Q_2, \dots\}$ obtained by Algorithm~\ref{alg:gb_construction}, $\ell_m$ needs to be an integer multiple of $\ell$, otherwise the dimension $k_m$ of code $\Q_m$ can be 0.

To see this, consider the construction where $p^{(m)}(x) = 1$ for all $m=1,\dots,M$ and set $\ell_m = \kappa_m\ell + r$ for some $m, \kappa_m, r \in \NN$ such that $0 < r < \ell$.
Then, 
\begin{equation}
    x^{m\ell + r} - 1 = x^r\left(\sum_{i=0}^{m-1} x^{i\ell} \right) (x^{\ell} - 1) + (x^r - 1),
\end{equation}
which is not a multiple of $(x^{\ell} - 1)$ due to the additional term $(x^r - 1)$ because $0 < r < \ell$.

For example, consider extending polynomials $a(x), b(x) \in \FF_2^{\langle 7 \rangle}$ given by $a(x) = 1 + x + x^3$ and $b(x) = 1 + x^2 + x^3 + x^4$ to polynomials $a'(x), b'(x) \in \FF_2^{\langle 8 \rangle}$ such that $a'(x) = a(x)$ and $b'(x) = b(x)$.
Polynomials $a(x), b(x)$ define a $[[14,6]]$ GB code, because $2\deg \gcd(a(x), b(x), x^{7}-1) = 6$.
However, we observe that 
\begin{align*}
    \gcd(a'(x), b'(x), x^{8}-1) = \gcd(a(x), b(x), x^{8}-1) = 1,
\end{align*}
giving a zero-dimensional code when we try to extend the code from $\FF_2^{\langle 7 \rangle}$ to $\FF_2^{\langle 8 \rangle}$.

\subsection{Sparsity of the generalized Shor Code}
\label{app:sparsity_shor_code}

The Shor quantum error correcting-code was initially created to correct all types of single qubit errors. It is given by the concatenation of the distance three phase-flip QEC code with the distance 3 bit-flip QEC code. Now, we consider the generalized $[[d^2,1,d]]$ Shor code given by the concatenation of distance $d$ phase-flip QEC code with the distance $d$ bit-flip QEC code. One can see that the stabilizer set is generated by $\mathcal{S} = \langle Z_{d\cdot i+1} Z_{d\cdot i+2},\ldots,Z_{d\cdot i+d-1} Z_{d\cdot i+d}, X_{d\cdot j+1} \cdots X_{d\cdot j+2d}| i=0,\ldots, d-1, \text{ and }j=0,\ldots,d-2\rangle.$ Since each $X$--stabilizer has weight $2d$ and each qubit is in the support of at maximum 2 $X$--stabilizers and 2 $Z$--stabilizers, we have that the row weight and column weight of the parity-check matrix $H_\text{Shor}$ of the generalized Shor code is $w_r(d) = 2d$ and $w_c(d) =4$, respectively. Thus, the sparsity of $H_{\text{Shor}}$, according to Definition~\ref{def:sparsity}, is equal to 
    \begin{eqnarray}
        \frac{w_r(d)}{n_d} &=& \frac{2d}{d^2}=\frac{2}{d},\\
        \frac{w_c(d)}{n_d} &=& \frac{4}{d^2} < \frac{2}{d}.
    \end{eqnarray}It is therefore evident that the family of generalized Shor codes $\F_{\text{Shor}}$ is $\mathcal{O}(1/m)$--sparse, where $m = (d-1)/2$ labels the elements in the family.

\subsection{Proof of Theorem~\ref{thm:scalable}}
\label{app:scalability}

\mainclaim*
\begin{proof}
    We first define an integer sequence and a polynomial sequence that allow us to follow the construction of Algorithm~\ref{alg:gb_construction} with $\Q_1$ as the base code.
    
    To this end, we define an integer sequence $(\kappa_1, \kappa_2, \dots)$ by 
    \begin{equation}
        \kappa_m = 3^{m-1} \,,
    \end{equation}
    so that $\ell_m = \kappa_m\ell_1$ for all $m = 1,\dots,M$.
    We then define a polynomial sequence $(p^{(1)}, p^{(2)}, \dots)$ by
    \begin{equation}
        p^{(1)} = 1 \text{ and } p^{(m)} = \prod_{k=1}^{m-1} \left(1+x^{\ell_k}\right) \,,
    \end{equation}
    for $m \geq 2$.
    We note that $p^{(m)} \in \FF_2^{\langle (\kappa_m - 1)\ell_1 + 1 \rangle}$, as
    \begin{equation}
       \deg p^{(m)} = \sum_{k=1}^{m-1} 3^{k-1} \ell_1 = \frac{1}{2} (3^{m-1}-1) \ell_1 < (\kappa_m - 1)\ell_1 + 1.
    \end{equation}
    
    Each code $\Q_m \in \F$ is therefore a GB code obtained from the construction of Algorithm~\ref{alg:gb_construction}.
    We now proceed to prove each statement in the theorem in turn.
    
    \begin{enumerate}
    \item Each code $\Q_m \in \F$ is a $[[2\ell_m, k_m]]$ code, where
    \begin{equation}\label{eq:ellm_scalable}
        \ell_{m+1} = 3\ell_m,
    \end{equation}
    for all $m \geq 1$.
    Therefore, $n_{m+1} = 3n_m$.
    This construction also leads to a stronger bound on the code dimension $k_m$ than the guarantee provided in Theorem~\ref{thm:guarantee}.
    To this end, we note that
    \begin{equation}
        x^{\ell_{m+1}} - 1 = (x^{\ell_m} - 1)(x^{2\ell_m} + x^{\ell_m} + 1),
    \end{equation}
    and
    \begin{equation}\label{eq:pm_scalable}
        p^{(m+1)} = (1+x^{3^{m-1}\ell_1})p^{(m)},
    \end{equation}
    for all $m \geq 1$, so
    \begin{align}\label{eq:ab_scalable}
        &\gcd\left(a^{(m+1)}(x),b^{(m+1)}(x)\right) \nonumber\\
        = &(1+x^{3^{m-1}\ell_1})\gcd\left(a^{(m)}(x),b^{(m)}(x)\right).
    \end{align}
    Therefore, we obtain
    \begin{align}
        k_{m+1} &= 2\hspace{2pt}\deg \gcd\left( a^{(m+1)}(x), b^{(m+1)}(x), x^{\ell_{m+1}} - 1) \right) \nonumber\\
        &\geq 2\hspace{2pt}\deg \gcd\left( \gcd\left(a^{(m)}(x), b^{(m)}(x)\right), x^{\ell_{m+1}} - 1 \right) \nonumber\\
        &\geq 2\hspace{2pt}\deg \gcd\big( a^{(m)}(x),b^{(m)}(x), x^{\ell_m} - 1 \big) \nonumber\\
        &= k_m,
    \end{align}
    where in the first inequality we have used \eqref{eq:ab_scalable} along with Proposition~\ref{thm:gcd_inequality}, and in the second inequality we have used \eqref{eq:ellm_scalable} along with Proposition~\ref{thm:gcd_inequality}.
    This proves statement~\ref{en:parameters}.
    
    \item We can consider the effect of polynomial $p^{(m)}$ given in \eqref{eq:pm_scalable} that transforms the code generating polynomials $a^{(m)},b^{(m)}$ on shaping the parity-check matrix $H_m$ inductively on $m$.
    We recall the isomorphisms $a^{(m)} \sim A_m$ and $b^{(m)} \sim B_m$ for the $m$-th code $\Q_m$. 
    We further use $a(x), b(x)$ to denote $a^{(1)}(x), b^{(1)}(x)$ respectively.
    
    For $m=1$, we retrieve the base code $\Q_1$ with parity-check matrix
    \begin{equation}
        H_1 \coloneqq \begin{pmatrix} (A_1|B_1) & 0 \\ 0 & (B_1^T|A_1^T) \end{pmatrix}.
    \end{equation}
    Moreover, we see that $p^{(2)}(x)$ maps $\textcolor{blue}{a(x)} \mapsto a^{(2)}(x) = \textcolor{blue}{a(x)} + x^{\ell_1}\textcolor{red}{a(x)}$, therefore it creates a new polynomial of weight $2\wt a(x)$ which is the concatenation of two copies of the original polynomial so that the circulant property of $A_2$ will ensure that $A_1$ is embedded in $A_2$.
    We construct $A_2$ by inserting $\textcolor{blue}{a(x)} + x^{\ell_1}\textcolor{red}{a(x)} \in \FF_2^{\langle 3\ell_1 \rangle}$ as its first column and building the remaining columns by shifting the first column one step downwards (following the definition of a circulant matrix in \eqref{eq:circulant}).
    The monomial $x^{\ell_1}$ represents a shift of the coefficients of polynomial $a(x)$ by $\ell_1$ spaces, and we color the two copies of the polynomial $a(x)$ to guide the reader,
    \begin{eqnarray}
        A_1 = 
        \begin{pmatrix}
            \textcolor{blue}{a_0} & \textcolor{blue}{a_{\ell-1}} & \dots & \textcolor{blue}{a_1} \\ 
            \textcolor{blue}{a_1} & \textcolor{blue}{a_0} & \dots & \textcolor{blue}{a_2} \\
            \dots & \dots & \dots & \dots \\
            \textcolor{blue}{a_{\ell-1}} & \textcolor{blue}{a_{\ell-2}} & \dots & \textcolor{blue}{a_0}
        \end{pmatrix}
        \mapsto 
        \tiny
        \begin{pmatrix}
            \textcolor{blue}{a_0} & 0 & \dots & 0 &\hspace{5pt} 0 & \textcolor{red}{a_{\ell-1}} & \dots & \textcolor{red}{a_1} &\hspace{5pt} \textcolor{red}{a_0} & \textcolor{blue}{a_{\ell-1}} & \dots & \textcolor{blue}{a_1} \\ 
            \textcolor{blue}{a_1} & \textcolor{blue}{a_0} & \dots & 0 &\hspace{5pt} 0 & 0 & \dots & \textcolor{red}{a_2} &\hspace{5pt} \textcolor{red}{a_1} & \textcolor{red}{a_0} & \dots & \textcolor{blue}{a_2} \\
            \dots & \dots & \dots & \dots &\hspace{5pt} \dots & \dots & \dots & \dots &\hspace{5pt} \dots & \dots & \dots & \dots \\
            \textcolor{blue}{a_{\ell-1}} & \textcolor{blue}{a_{\ell-2}} & \dots & \textcolor{blue}{a_0} &\hspace{5pt} 0 & 0 & \dots & 0 &\hspace{5pt} \textcolor{red}{a_{\ell-1}} & \textcolor{red}{a_{\ell-2}} & \dots & \textcolor{red}{a_0}
            \\[10pt]
            \textcolor{red}{a_0} & \textcolor{blue}{a_{\ell-1}} & \dots & \textcolor{blue}{a_1} &\hspace{5pt} \textcolor{blue}{a_0} & 0 & \dots & 0 &\hspace{5pt} 0 & \textcolor{red}{a_{\ell-1}} & \dots & \textcolor{red}{a_1} \\ 
            \textcolor{red}{a_1} & \textcolor{red}{a_0} & \dots & \textcolor{blue}{a_2} &\hspace{5pt} \textcolor{blue}{a_1} & \textcolor{blue}{a_0} & \dots & 0 &\hspace{5pt} 0 & 0 & \dots & \textcolor{red}{a_2} \\
            \dots & \dots & \dots & \dots &\hspace{5pt} \dots & \dots & \dots & \dots &\hspace{5pt} \dots & \dots & \dots & \dots \\
            \textcolor{red}{a_{\ell-1}} & \textcolor{red}{a_{\ell-2}} & \dots & \textcolor{red}{a_0}  &\hspace{5pt} \textcolor{blue}{a_{\ell-1}} & \textcolor{blue}{a_{\ell-2}} & \dots & \textcolor{blue}{a_0} &\hspace{5pt} 0 & 0 & \dots & 0
            \\[10pt]
            0 & \textcolor{red}{a_{\ell-1}} & \dots & \textcolor{red}{a_1} &\hspace{5pt} \textcolor{red}{a_0} & \textcolor{blue}{a_{\ell-1}} & \dots & \textcolor{blue}{a_1} &\hspace{5pt} \textcolor{blue}{a_0} & 0 & \dots & 0 \\ 
            0 & 0 & \dots & \textcolor{red}{a_2} &\hspace{5pt} \textcolor{red}{a_1} & \textcolor{red}{a_0} & \dots & \textcolor{blue}{a_2} &\hspace{5pt} \textcolor{blue}{a_1} & \textcolor{blue}{a_0} & \dots & 0 \\
            \dots & \dots& \dots & \dots &\hspace{5pt} \dots & \dots & \dots & \dots &\hspace{5pt} \dots & \dots & \dots & \dots \\
            0 & 0 & \dots & 0 &\hspace{5pt} \textcolor{red}{a_{\ell-1}} & \textcolor{red}{a_{\ell-2}} & \dots & \textcolor{red}{a_0} &\hspace{5pt} \textcolor{blue}{a_{\ell-1}} & \textcolor{blue}{a_{\ell-2}} & \dots & \textcolor{blue}{a_0}
        \end{pmatrix}.
    \end{eqnarray}Now, it is clear that the last matrix is equal to
    \begin{eqnarray}
        \begin{pmatrix} L(A_1) & U(A_1) & A_1 \\ A_1 & L(A_1) & U(A_1) \\ U(A_1) & A_1 & L(A_1) \end{pmatrix} = A_2. \nonumber
        \nonumber
    \end{eqnarray}
    The final matrix is written in block form with three blocks equal to $A_1$, three blocks equal to its lower triangular form $L(A_1)$ and three blocks equal to its upper triangular form $U(A_1)$.
    Matrix $B_2$ takes a completely analogous form in terms of $B_1$ due to the mapping $\textcolor{blue}{b(x)} \mapsto b^{(2)}(x) = \textcolor{blue}{b(x)} + x^{\ell_1}\textcolor{red}{b(x)} \in \FF_2^{\langle 3\ell \rangle}$.
    
    For any $m \geq 2$, we obtain the generating polynomial $a^{(m+1)}(x)$ by
    \begin{equation}
        a^{(m)}(x) \mapsto a^{(m+1)}(x) = p^{(m)}(x)a(x) = a^{(m)}(x) + x^{\ell_{m}} a^{(m)}(x),
    \end{equation}
    Since $\ell_m = 3\ell_{m-1}$, this is equivalent to the mapping
    \begin{align}
        A_m \mapsto A_{m+1} = \begin{pmatrix} L_m & U_m & A_m \\ A_m & L_m & U_m \\ U_m & A_m & L_m \end{pmatrix} = F(A_m) \,,
    \end{align}
    where $L_m \coloneqq L(A_m)$ and $U_m \coloneqq U(A_m)$ are respectively the lower and upper triangular matrices obtained from $A_m$.
    We can set $b^{(m)}(x) \mapsto b^{(m+1)}(x)$ and $B_m \mapsto B_{m+1}$ analogously.
    We have thus shown statement~\ref{en:Hm} of the theorem.
    
    \item Matrices $A_m, B_m$ are both $(\ell_m \times \ell_m)$ matrices of row and column weight equal to $2^{m-1}\hspace{1pt}\wt a(x)$ and $2^{m-1}\hspace{1pt}\wt b(x)$ respectively. 
    The weights double at each iterations as we are concatenating two copies of the $m$-th generating polynomials $a^{(m)}(x), b^{(m)}(x)$ to construct the $(m+1)$-th circulant matrices $A_{m+1}, B_{m+1}$.
    Therefore, the $(2\ell_m \times 2n_m)$ parity-check matrix  \begin{equation}
        H_m \coloneqq \begin{pmatrix} (A_m|B_m) & 0 \\ 0 & (B_m^T|A_m^T) \end{pmatrix},
    \end{equation}
    
    has row weight $w_r(m) \coloneqq 2^{m-1}\hspace{1pt} ( \wt a(x) + \wt b(x) )$ and column weight $w_c(m) \coloneqq 2^{m-1}\hspace{1pt} \max\{\wt a(x), \wt b(x)\}$.
    To calculate the sparsity of $H_m$ according to Definition~\ref{def:sparsity}, we observe that
    \begin{align}
        \frac{w_r(m)}{n_m} = \frac{2^{m-1} ( \wt a(x) + \wt b(x) )}{2 \cdot 3^{m-1} \ell_1} \leq \frac{2^{m-1} \max\{ \wt a(x), \wt b(x) \}}{3^{m-1} \ell_1} = \frac{w_c(m)}{\ell_m},
    \end{align}
    where the inequality follows from the fact that $\max\{ c_1, c_2 \} \geq \frac{c_1 + c_2}{2}$ for any real numbers $c_1,c_2$.
    It is therefore evident that $\F$ is a family of $q(m)$--sparse codes, where $q(m) = \frac{w_c(m)}{\ell_m}$, proving statement~\ref{en:sparsity}.
    
    \item To prove that $\F$ is scalable according to Definition~\ref{def:scalability}, we need to identify an appropriate qubit relabelling such that $H_m$ is exactly embedded in $H_{m+1}$.
    
    We first introduce some simplified notation.
    Denote by $H_X, H_Z$ the $X$ and $Z$ parts of $H_m$ respectively and denote by $H'_X, H'_Z$ the $X$ and $Z$ parts of $H_{m+1}$ respectively.
    Further let $L_m \coloneqq L(A_m)$, $U_m \coloneqq U(A_m)$ and $L'_m \coloneqq L(B_m)$ and $U'_m \coloneqq U(B_m)$.
    
    We begin by showing how to embed $H_X = (A_m|B_m)$ into $H'_X = (A_{m+1}|B_{m+1})$.
    By writing out $H'_X$ explicitly, we can follow a sequence of qubit relabellings to get
    \begin{align}
        H_X' &= 
        \begin{pmatrix} 
            L_m & U_m & A_m & L'_m & U'_m & B_m \\ 
            A_m & L_m & U_m & B_m & L'_m & U'_m \\ 
            U_m & A_m & L_m & U'_m & B_m & L'_m 
        \end{pmatrix}
        \mapsto 
        \begin{pmatrix} 
            A_m & B_m & L_m & L'_m & U_m & U'_m \\ 
            U_m & U'_m & A_m & B_m & L_m & L'_m \\ 
            L_m & L'_m & U_m & U'_m & A_m & B_m
        \end{pmatrix}
        \nonumber\\
        &= 
        \begin{pmatrix} 
            (A_m | B_m) & \big(L_m | L'_m\big) & \big(U_m | U'_m\big) \\ 
            \big(U_m | U'_m\big) & (A_m | B_m) & \big(L_m | L'_m\big) \\ 
            \big(L_m | L'_m\big)& \big(U_m | U'_m\big) & (A_m | B_m)
        \end{pmatrix}
        .
    \end{align}
    The first step follows by relabelling the block columns according to the permutation
    $\bigl(\begin{smallmatrix}
  1 & 2 & 3 & 4 & 5 & 6 \\
  3 & 5 & 1 & 4 & 6 & 2
\end{smallmatrix}\bigr),$
    which corresponds to a data qubit relabelling.
    Now $H_X$ is embedded in the top left corner of $H'_X$ satisfying the condition of Definition~\ref{def:scalability} for scalability.
    
    We can repeat this process to ensure that $H_Z = (B_m^T|A_m^T)$ is embedded in $H_Z' = (B_{m+1}^T|A_{m+1}^T)$, 
    \begin{align}
        H_Z' = 
        \begin{pmatrix} 
            L^{\prime T}_m & B_m^T & U^{\prime T}_m & L_m^T & A_m^T & U_m^T \\ 
            U^{\prime T}_m & L^{\prime T}_m & B_m^T & U_m^T & L_m^T & A_m^T \\ 
            B_m^T & U^{\prime T}_m & L^{\prime T}_m & A_m^T & U_m^T & L_m^T
        \end{pmatrix}
        \mapsto 
        \begin{pmatrix} 
            (B_m^T | A_m^T) & \big(L^{\prime T}_m | L_m^T\big) & \big(U^{\prime T}_m | U_m^T\big) \\ 
            \big(U^{\prime T}_m | U_m^T\big) & (B_m^T | A_m^T) & \big(L^{\prime T}_m | L_m^T\big) \\ 
            \big(L^{\prime T}_m | L_m^T\big) & \big(U^{\prime T}_m | U_m^T\big) & (B_m^T | A_m^T)
        \end{pmatrix}
        .
    \end{align}
    where now the qubit relabelling corresponds to the permutation $\bigl(\begin{smallmatrix}
  1 & 2 & 3 & 4 & 5 & 6 \\
  3 & 1 & 5 & 4 & 2 & 6
\end{smallmatrix}\bigr)$ of block columns in $H_Z'$ followed by swapping the second and third rows of blocks.
    \end{enumerate}
\end{proof}

\subsection{Proof of Theorem~\ref{thm:scalable2}}
\label{app:scalability2}

\secondmainclaim*
\begin{proof}
    First of all, we need to show the existence of an integer and a polynomial sequence that can be used to construct this family of codes following Algorithm~\ref{alg:gb_construction} with $\Q_1$ as the base code. To this end, consider the integer sequence $(\kappa_1, \kappa_2, \dots)$ given by 
    \begin{equation}
        \kappa_m = \frac{\ell + r\cdot m}{\ell} \,,
    \end{equation}
    so that $\ell_m = \kappa_m\ell_1=\ell + r\cdot m$, for $m = 1,\dots,M$.
    Let $a(x)=f^{(a)}(x) + g^{(a)}(x)$ and $b(x)=f^{(b)}(x) + g^{(b)}(x)$, with $f^{(a)}(x)=\sum_{t=0}^{j - 1}a_t x^t$, $g^{(a)}(x) = \sum_{t=j}^{\ell - 1}a_t x^t$, $f^{(b)}(x)=\sum_{t=0}^{j - 1}b_t x^t$, and $g^{(b)}(x) = \sum_{t=j}^{\ell - 1}b_t x^t$,  be polynomials isomorphic to the matrices $A_1$ and $B_1$, respectively. We then define a polynomial sequence $(p^{(1)}, p^{(2)}, \ldots, p^{(M)})$ by imposing the following set of equalities
    \begin{eqnarray}\label{eq:pm_scalable2}
        \sum_{t=0}^s \gamma_t p_{s-t}^{(m)} &=& \gamma_s, \text{ for } 0\leq s\leq j-1,\nonumber\\
        \sum_{t=0}^s \gamma_t p_{s-t}^{(m)} &=& 0, \text{ for } j\leq s\leq j+r\cdot m-1,\nonumber\\
        \sum_{t=0}^s \gamma_t p_{s-t}^{(m)} &=& \gamma_s, \text{ for } j+r\cdot m\leq s\leq \ell + r\cdot m-1,
    \end{eqnarray}for $\gamma = a, b$ and $m = 1,\dots,M$, where $p^{(m)}(x)=\sum_{t=0}^{r\cdot m} p_t^{(m)}x^t$. From construction, it is clear that $\deg a + \deg p^{(m)}, \deg b + \deg p^{(m)} < \ell + r\cdot m$.
    Similarly to the previous theorem, each code $\Q_m \in \F$ is therefore a GB code obtained from the construction of Algorithm~\ref{alg:gb_construction}.
    Proceeding to the proof of the above statements of the current theorem, we have the following:
    
    \begin{enumerate}
    \item Each code $\Q_m \in \F$ is a $[[2\ell_m, k_m]]$ code, where
    \begin{equation}\label{eq:ellm_scalable}
        \ell_{m+1} = \ell_m + r,
    \end{equation}
    for all $m \geq 1$.
    Therefore, $n_{m+1} = n_m+2r$. 
    The bound on the code dimension is obtain from Theorem~\ref{thm:guarantee}.
    
    \item In the following, we are going to evaluate the action of the proposed polynomial $p^{(m)}$ given in \eqref{eq:pm_scalable2} on the polynomials $a(x), b(x)$ defining the code generating polynomials $a^{(m)},b^{(m)}$.
    As before, we have that $a^{(m)} \sim A_m$ and $b^{(m)} \sim B_m$ for the $m$-th code $\Q_m$, and $a(x), b(x)$ denote $a^{(1)}(x), b^{(1)}(x),$ respectively.
    
    For $m=1$, we retrieve the base code $\Q_1$ with parity-check matrix
    \begin{equation}
        H_1 \coloneqq \begin{pmatrix} (A_1|B_1) & 0 \\ 0 & (B_1^T|A_1^T) \end{pmatrix}.
    \end{equation}
    Moreover, we see that $p^{(m)}(x)$ maps $\textcolor{blue}{a(x)} \mapsto a^{(m)}(x) = f^{(a)}(x) + x^{r\cdot m}g^{(a)}(x)$, therefore it creates a new polynomial of the same weight as $a(x)$. Describing the new matrix $A_m$ by the elements of the matrix $A_1$, we have
    \begin{eqnarray}\label{Eq:mapstothirdconstruction}
        \hspace{-.8cm}A_1 = 
        \begin{pmatrix}
            \textcolor{blue}{a_0}       & \cdots & \textcolor{blue}{a_{j-1}}  & \cdots & \textcolor{blue}{a_{1}}\\ 
            \dots                       & \dots  & \dots                      & \dots  & \dots\\
            \textcolor{blue}{a_{j-1}}   & \cdots & \textcolor{red}{a_{2j-1}}  & \cdots & \textcolor{blue}{a_{j-2}}\\ 
            \textcolor{red}{a_{j}}      & \cdots & \textcolor{red}{a_{2j}}    & \cdots & \textcolor{blue}{a_{j-1}}\\ 
            \dots                       & \dots  & \dots                      & \dots  & \dots\\
            \textcolor{red}{a_{\ell-1}} & \cdots & \textcolor{blue}{a_{j-2}}  & \cdots & \textcolor{blue}{a_{0}}
        \end{pmatrix}
        \mapsto 
        \tiny
        \begin{pmatrix}
            \textcolor{blue}{a_0} & \textcolor{red}{a_{\ell-1}} & \dots & \textcolor{red}{a_{\ell - j + 1}} & \textcolor{red}{a_{\ell - j}} & \dots & \textcolor{red}{a_{j}} & 0 & \dots & \textcolor{blue}{a_1} \\ 
            \textcolor{blue}{a_1} & \textcolor{blue}{a_0} & \dots & \textcolor{red}{a_{\ell - j + 2 }} & \textcolor{red}{a_{\ell - j + 1}} & \dots & \textcolor{red}{a_{j + 1}} & \textcolor{red}{a_j} & \dots & \textcolor{blue}{a_2} \\
            \dots & \dots & \dots & \dots & \dots & \dots & \dots & \dots & \dots & \dots \\
            \textcolor{blue}{a_{j-1}} & \textcolor{blue}{a_{j-2}} & \dots & \textcolor{blue}{a_0} & \textcolor{red}{a_{\ell-1}} & \dots & \textcolor{red}{a_{2j - 1}} & \textcolor{red}{a_{2j - 2}} & \dots & \textcolor{blue}{a_{j-1}} \\
            0 & \textcolor{blue}{a_{j-1}} & \dots & \textcolor{blue}{a_1} & \textcolor{blue}{a_0} & \dots & \textcolor{red}{a_{2j}} & \textcolor{red}{a_{2j-1}} & \dots & 0 \\ 
            0 & 0 & \dots & \textcolor{blue}{a_2} & \textcolor{blue}{a_1} & \dots & \textcolor{red}{a_{2j+1}} & \textcolor{red}{a_{2j}} & \dots & 0 \\
            \dots & \dots & \dots & \dots & \dots & \dots & \dots & \dots & \dots & \dots \\
            0 & 0 & \dots & 0  & \textcolor{blue}{a_{j-1}} & \dots & \textcolor{blue}{a_0} & \textcolor{red}{a_{\ell - 1}} & \dots & \textcolor{red}{a_{j}}\\
            \textcolor{red}{a_j} & 0 & \dots & 0 & 0 & \dots & \textcolor{blue}{a_1} & \textcolor{blue}{a_0} & \dots & \textcolor{red}{a_{j+1}} \\ 
            \textcolor{red}{a_{j+1}} & \textcolor{red}{a_{j}} & \dots & 0 & 0 & \dots & \textcolor{blue}{a_2} & \textcolor{blue}{a_1} & \dots & \textcolor{red}{a_{j+2}} \\
            \dots & \dots& \dots & \dots & \dots & \dots & \dots & \dots & \dots & \dots \\
            \textcolor{red}{a_{\ell-1}} & \textcolor{red}{a_{\ell-2}} & \dots & \textcolor{red}{a_{\ell - j}} & \textcolor{red}{a_{\ell- j - 1}} & \dots & 0 & 0 & \dots & \textcolor{blue}{a_0}
        \end{pmatrix},
    \end{eqnarray}where the elements $a_0,\ldots,a_{j-1}$ have been colored in blue and $a_j,\ldots,a_{\ell-1}$ in red to guide the reader. Now, it is clear that the last matrix is equal to Eq.~(\ref{eq:pc_inductive2}). The same argument can be used for the matrices $B_m$, for $m=1,\ldots, M$.
    
    \item Matrices $A_m$ and $B_m$ are both $(\ell_m \times \ell_m)$ matrices with row and column weights equal to $\wt a(x)$ and $\wt b(x)$, respectively. Observe that these weights are kept the same for any code element in the family. Therefore, the $(2\ell_m \times 2n_m)$ parity-check matrix
    \begin{equation}
        H_m \coloneqq \begin{pmatrix} (A_m|B_m) & 0 \\ 0 & (B_m^T|A_m^T) \end{pmatrix},
    \end{equation} has row weight $w_r(m) = \wt a(x) + \wt b(x)$ and column weight $w_c(m) = \max\{\wt a(x), \wt b(x)\}$. Since both weights do not depend on which element of the family we are considering and $\wt a(x) + \wt b(x) > \max\{\wt a(x), \wt b(x)\}$, it is clear that $\F$ is a family of $t$-qLDPC codes with $t = \wt a(x) + \wt b(x).$
    \end{enumerate}
\end{proof}

\section{Numerical performance of codes}
\label{app:numerics}

\begin{figure}[t]
    \centering 
    \includegraphics[width=1\columnwidth]{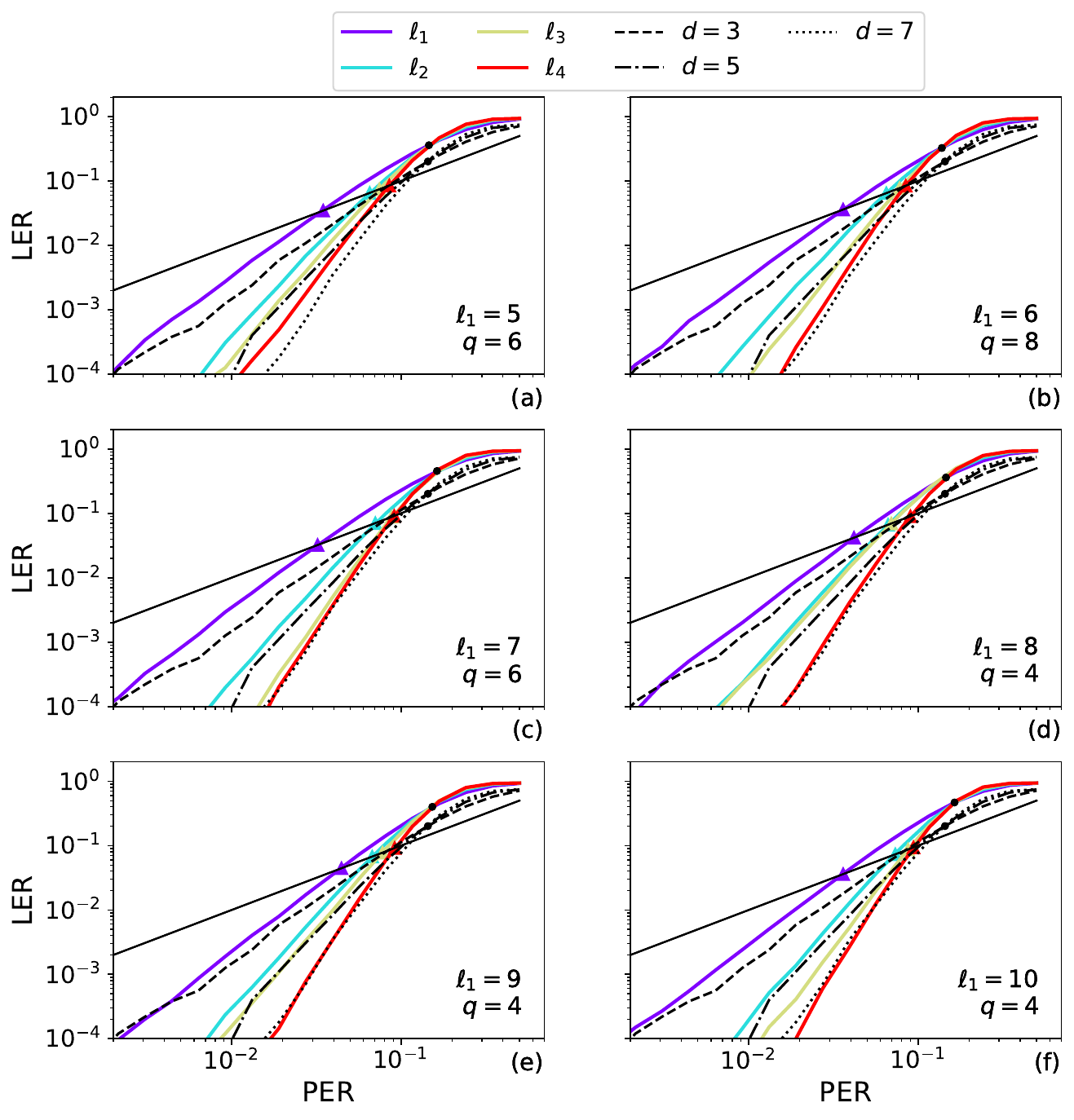}
    \caption{ {\bf qLDPC families for different base code lengths.}
    Each plot displays a qLDPC code family with given $[[\ell_1, k_1]]$ base code $\ell_1$ for $\ell_1 = 5,6,7,8,9,10$.
    The base code is defined by polynomials $a(x),b(x) \in \FF_2^{\langle \ell_1 \rangle}$ provided in the text and the code dimension is given by $k_1 = 2\hspace{2pt}\deg g(x)$ according to Proposition~\ref{prop:k}.
    For each code family, we plot the smallest four codes, showing the breakeven point of each code and the threshold point. 
    We compare them with the surface code of distances $3,5,7$.
    }
    \label{fig:appendix}
\end{figure}

Here we present qLDPC codes of comparable performance for all $\ell_1 = 5, \dots, 10$.
For a given pair of generators $a(x),b(x) \in \FF_2^{\langle \ell_1 \rangle}$, we construct the base code and extend it according to Algorithm~\ref{alg:gb_construction} using the integer sequence $\kappa_m = m$ and $p^{(m)} = 1$ for $m \geq 1$.
Therefore, we obtain small qLDPC codes, which we plot in \figref{fig:appendix} up to $m=4$.
We identify each code family by the size $\ell_1$ of its base code. 
The six base codes that we use to create the families in \figref{fig:appendix} are:
\begin{enumerate}
    \item[(a)] $[[10,2,3]]$ code $(\ell_1 = 5)$; plotted in \figref{fig:threshold} too.
    \begin{align*}
        a(x) = 1 + x^4 \,; \quad
        b(x) = 1 + x + x^2 + x^4 \,.
    \end{align*}
    \item[(b)] $[[12,2,3]]$ code $(\ell_1 = 6)$.
    \begin{align*}
        a(x) = 1 + x + x^2 + x^5 \,; \quad
        b(x) = 1 + x + x^3 + x^5 \,.
    \end{align*}
    \item[(c)] $[[14,2]]$ code $(\ell_1 = 7)$.
    \begin{align*}
        a(x) = 1 + x^3 \,; \quad
        b(x) = 1 + x + x^3 + x^6 \,.
    \end{align*}
    \item[(d)] $[[16,2]]$ code $(\ell_1 = 8)$.
    \begin{align*}
        a(x) = x + x^3 \,; \quad
        b(x) = 1 + x^5 \,.
    \end{align*}
    \item[(e)] $[[18,2]]$ code $(\ell_1 = 9)$.
    \begin{align*}
        a(x) = 1 + x^2 \,; \quad
        b(x) = 1 + x^5 \,.
    \end{align*}
    \item[(f)] $[[20,2]]$ code $(\ell_1 = 10)$.
    \begin{align*}
        a(x) = 1 + x \,; \quad
        b(x) = 1 + x^6 \,.
    \end{align*}
\end{enumerate}

Note that all base codes we have selected have $g(x) = 1 + x$, thereby code dimension $k_1 = 2$.

In the plots, we include the threshold point of each code family along with the breakeven points, i.e.~the point at which LER dips below PER for a given code. 
We observe that the breakeven points increase as $m$ increases for each code family, and the threshold of each family is similar to the surface code threshold.
Bigger codes, i.e.~higher values of $\ell_1$, allow for a small number of qubits per parity check $q \coloneqq \wt a(x) + \wt b(x)$.

\end{document}